\documentclass[10pt]{article}
\usepackage{subcaption}
\usepackage{amsfonts,amssymb,amscd,amsmath,enumerate,verbatim,calc}
\usepackage{float}
\usepackage{graphics}
\usepackage{graphicx}
\usepackage{cite}
\usepackage{latexsym}
\usepackage{amsmath}
\usepackage{amssymb}
\usepackage{amsthm}

\usepackage{subcaption}

\usepackage{graphicx}

\theoremstyle{plain}
\newtheorem{theorem}{Theorem}[section]
\newtheorem{proposition}{Proposition}[section]
\newtheorem{lemma}{Lemma}[section]

\theoremstyle{definition}
\newtheorem{definition}{Definition}[section]

\theoremstyle{remark}
\newtheorem{remark}[theorem]{Remark}

\begin{document}
\begin{center}
	{\bf\Large A framework of windowed octonion linear canonical transform}
\end{center}
  
\begin{center}
	{\bf Manish Kumar* and Bhawna}\\
	\author{Manish Kumar}
	{Department of Mathematics, Birla Institute of Technology and 
		Science-Pilani, Hyderabad Campus, Hyderabad-500078, Telangana, India}\\
	{*Corresponding author}: {manishkumar@hyderabad.bits-pilani.ac.in}
	
\end{center}

\begin{abstract}
  The uncertainty principle is a fundamental principle in theoretical physics, such as quantum mechanics and classical mechanics. It plays a prime role in signal processing, including optics, where a signal is to be analyzed simultaneously in both domains; for instance, in harmonic analysis, both time and frequency domains, and in quantum mechanics, both time and momentum. On the other hand, many mathematicians, physicists, and other related domain researchers have paid more attention to the octonion-related integral transforms in recent years. In this paper, we define important properties of the windowed octonion linear canonical transform (WOCLCT), such as inversion, linearity, parity, shifting, and the relationship between OCLCT and WOCLCT. Further, we derived sharp Pitt's and sharp Young-Hausdorff inequalities for 3D WOCLCT. We obtain the logarithmic uncertainty principle for the 3D WOCLCT. Furthermore, Heisenberg's and Donoho-Stark's uncertainty principles are derived for WOCLCT, and the potential applications of WOCLCT are also discussed.
\end{abstract}
\section{Introduction}
Many interesting physical and engineering systems are characterized by a wide range of multi-channel signals (for instance, seismic signals have four channels, a color image has three channels, etc.). Sometimes these multi-channel signals with several components must be controlled simultaneously (for instance, image encryption see \cite{color,novel,pattern,image,hyper}). However, implementation becomes
challenging, especially for problems dealing with multi-channel signals. Taking each channel
at a time and considering its integral transform does not yield a desirable outcome. Applied mathematicians
and engineers encounter this problem in several applications of practical interest, such as structural design,
predicting earthquakes using seismic signals, computer graphics, aerospace engineering, quantum mechanics,
time-frequency analysis, optics, signal processing, image processing and enhancement, pattern recognition,
artificial intelligence, etc. On the other hand, one can see many real-life applications based on hyper-complex algebra-based transforms where multichannel components need to be processed simultaneously (see, for more details \cite{color,novel,two,image, pattern,hyper} and references therein). Motivated by lack of processing multi-channel signals simultaneously. The windowed octonion linear canonical transform (WOCLCT) appears to be a
promising method. \\
\par The WOCLCT is a family of integral transforms, which is a generalization of many integral transforms, including quaternion windowed linear canonical transform (QWLCT) \cite{akQWLCT2023}, quaternion linear canonical transform QLCT  \cite{1}, quaternion Fourier transform (QFT) \cite{2,bayrofourier2007,pbaswater2003}, the quaternion fractional Fourier transform (QFRFT) \cite{3}, the octonion Fourier transform (OFT) \cite{panoctonion2019}, windowed linear canonical transform (WLCT) \cite{WLCT} and many more. The WOCLCT is a generalization of the WLCT to octonions. Octonions are disordered, non-commutative, non-associative, alternative, no non-trivial zero divisors algebra of dimension eight that generalizes real numbers, complex numbers, and quaternions.
The WOCLCT has properties that make it useful for analyzing non-commutative systems, such as those found in quantum mechanics and relativistic physics. Uncertainty principles based on the Fourier transform can be viewed in \cite{havinuncer2012,follandmathematical1997}. Many works have been reported in the literature on developments of such integral transforms; see for more detail \cite{williampitt, akQWLCT2023,panoctonion2019,wentheoctonion2021}. An application, examples, and uncertainty principles (Donoho-Stark's inequality, Pitt's inequality, Heisenberg's inequality, Lieb, and local including reproducing kernel and characterization range) using quaternion window linear canonical transform (QWLCT) are researched in \cite{akQWLCT2023}. The authors established QFT and OFT relations and explored the Mustard convolution using the OFT, including several uncertainty principles in \cite{panoctonion2019}. The properties and uncertainty principles are derived using logarithmic estimates obtained from a sharp form of Pitt's inequality. Further few more results obtained on the Hardy-Littlewood-Sobolev inequality, including entropy using Fourier transform, are obtained in \cite{williampitt}. In \cite{wentheoctonion2021}, the authors explored important applications, properties such as isometry, shifting properties, inversion property, Riemann-Lebesgue Lemma, including uncertainty principle (such as Heisenberg's inequality and Donoho-Stark's inequality) are obtained using OCLCT.  \\
\par To the best of our knowledge, the theory of WOCLCT has not yet appeared in the literature and is still an open area for researchers. Motivated by the fact that WOCLCT is a new area of research, we contribute first in establishing important properties of WOCLCT, such as inversion, linearity, parity, and shifting. Further, we derived the main inequality, such as sharp Pitt's and sharp Young-Hausdorff inequalities, including uncertainty principles (i.e., logarithmic uncertainty principle, Heisenberg's uncertainty principle, and Donoho-Stark's uncertainty principle) for 3D WOCLCT and the potential applications of WOCLCT are also discussed.
\section*{Organization of the work} In section \ref{sec:2}, we recall some basic properties and definitions of 3D OFT and its inverse. in this section, we also discussed the definition of 3D OCLCT and its inverse. In section \ref{sec:3}, we define 3D WOCLCT and its inverse, including important properties of WOCLCT, such as inversion, linearity, parity, and shifting. In section \ref{sec:4}, the work's main contribution in the direction of inequalities and uncertainty principles (sharp Pitt's inequality, sharp Young-Hausdorff inequality, logarithmic uncertainty principle, Heisenberg's and Donoho-Stark's uncertainty principle for 3D WOCLCT) are derived. Moreover, the potential applications of 3D WOCLCT are also discussed in section \ref{sec:5}. Finally, we conclude the work in section \ref{sec:6}.
\section{Preliminaries}\label{sec:2}
In this section, we provide basic information on octonion algebra, some basic definitions of OCLCT, and an important Lemma used throughout the work. In history \cite{book}, John T. Graves, along with Hamilton, discovered octonions and called them octaves but did not publish work until 1845. Arthur Cayley published his discovery of the octonions and provided a name to it as Cayley numbers. Hamilton reorganized that Cayley published first, but the invention of octonions was done before with Graves. Hence, credit to both ware given for independently discovering the octonions. The octonions are constructed through Cayley-Dickons Process as $\mathbb{O}=\mathbb{H}+\mathbb{H}e_{4}$.
\subsection*{Algebra on octonions} Let us consider standard natural basis set $\{e_k; \ k=0, 1,\dots, 7 \}$ to represent octonion numbers. For simplicity, we assume $e_{0}=1$ and renaming seven independent basis elements are imaginary units, and we could write for every $z \in \mathbb{O}$ as follows: 
\begin{align*}
z = z_0 + z_1e_1 +z_2e_2 +z_3e_3 +z_4e_4 +z_5e_5 +z_6e_6 +z_7e_7, 
\end{align*}
Further, the octonion conjugate is defined by:
\begin{align*}
\bar{z}=z_0 - z_1e_1 -z_2e_2 -z_3e_3 -z_4e_4 -z_5e_5 -z_6e_6 -z_7e_7,
\end{align*}
and satisfies
\begin{align*}
\overline{z_{1}z_{2}}=\bar{z_2} \bar{z_1}.
\end{align*}
The norm of an octonion $|z|$ is defined by:
\begin{align*}
|z|^2=z \bar{z} = z_{0}^{2}+z_{1}^{2}+z_{2}^{2}+z_{3}^{2}+z_{4}^{2}+z_{5}^{2}+z_{6}^{2}+z_{7}^{2},  
\end{align*}
for each $k=0, 1,\dots, 7,$ the components $z_k \in \mathbb{R}$, which can be thought of as a point or vector in $\mathbb{R}^{8}$. The real part of $z$ is just $z_1$; the imaginary part of $z$ is everything else. This is more similar to a complex number but slightly different from representing an imaginary part with seven degrees of freedom and can be viewed as a vector in $\mathbb{R}^{7}$. As we know, $e_1, e_2, \dots e_7$ are the seven imaginary standard units on octonion algebra which is non-associative and non-commutative over the field of real numbers. The multiplication table is provided in figure \ref{fig:1}; using a 7-point projective plane. Each point corresponds to an imaginary unit, and each line corresponds to a quaternionic triple, with the arrow giving orientation. As the other division algebras, the norm satisfies the identity $|z_1z_2|=|z_1||z_2| $. 
The octonion $z \in \mathbb{O}$ can also be represented as
\begin{align*}
z= \gamma+\delta e_4,
\end{align*}
where $ \gamma = z_0+ z_1e_1 +z_2e_2 +z_3e_3 $ and $ \delta=z_4+z_5e_1+z_6e_2+z_7e_3$ are quaternions, which belongs to $\mathbb{H}$. Now, we borrow the following Lemma provided in \cite{panoctonion2019}.
\begin{figure}
	\centering
	\includegraphics[scale=.5]{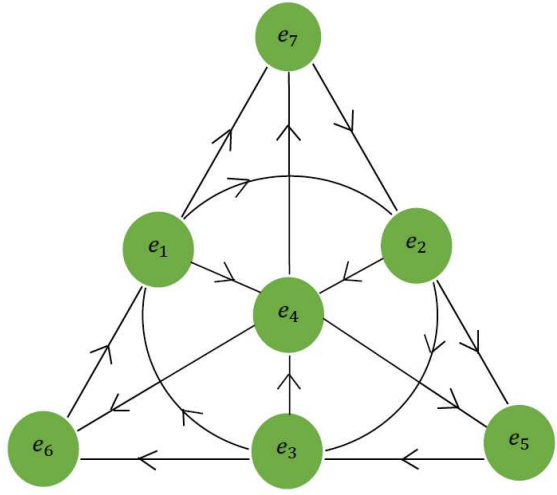}
	\caption{Octonion multiplication table.}
	\label{fig:1}
\end{figure}
\begin{lemma} Let $\gamma,\delta \in \mathbb{H}$, then
	\begin{align*}
	&(i) \hspace{0.2cm} e_{4}\gamma= \bar{\gamma}e_4;  \hspace{0.5cm} (ii) \hspace{0.2cm} e_4(\gamma e_4)=- \bar{\gamma}; \hspace{0.5cm} (iii)\hspace{0.2cm} (\gamma e_4)e_4= -\gamma;\\&(iv) \hspace{0.2cm} \gamma(\delta e_{4})= (\delta\gamma)e_4;\hspace{0.5cm} (v) \hspace{0.2cm} \gamma e_4(\delta)= (\gamma \bar{\delta})e_4; \hspace{0.5cm} (vi) \hspace{0.2cm} (\gamma e_4)(\delta e_4)=- \bar{\delta}\gamma.
	\end{align*}
\end{lemma}
From this Lemma, one can conclude that
\begin{align*}
\overline{\gamma + \delta e_4} = \bar{\gamma} - {\delta} e_4; \hspace{0.5cm}  |\gamma + \delta e_4 |^2= |\gamma|^2+|\delta|^2.
\end{align*}
An octonion-valued signal (or function) $f(t_1,t_2,t_3)$ is a map from $\mathbb{R}^3$ to $\mathbb{O}$ which takes the following explicit form as follows:
\begin{align*}
f(t_1,t_2,t_3)&=f_0(t_1,t_2,t_3)+f_{1}(t_1,t_2,t_3)e_{1}+f_{2}(t_1,t_2,t_3)e_{2}+f_{3}(t_1,t_2,t_3)e_{3}\\&+f_{4}(t_1,t_2,t_3)e_{4}+f_{5}(t_1,t_2,t_3)e_{5}+f_{6}(t_1,t_2,t_3)e_{6}+ f_{7}(t_1,t_2,t_3)e_{7},
\end{align*}
where each $f_k(t_1,t_2,t_3)$ is a real-valued signal (or function) for $k= 0,1, \dots ,7$. The $L^p$
norm $1 \leq p < \infty $, for each octonion-valued signal (or function) $f(t_1,t_2,t_3)$
over $\mathbb{R}^3$ 
is defined by:
\begin{align*}
||f(t_1,t_2,t_3)||_p:=\left(\int_{-\infty}^{\infty}\int_{-\infty}^{\infty}\int_{-\infty}^{\infty}\Big|f(t_1,t_2,t_3)\Big|^p{\rm d}t_1{\rm d}t_2{\rm d}t_3\right)^{\frac{1}{p}} \hspace{0.3cm} < \infty .
\end{align*}
The $L^{\infty}$ norm is given by
\begin{align*}
||f(t_1,t_2,t_3)||_{\infty}= ess \sup_{x \in \mathbb{R}^3}\Big|f(t_1,t_2,t_3)\Big|, \ \text{for} \ p=\infty.
\end{align*}
\begin{definition}[3D OFT] The 3D OFT of an octonion-valued signal (or function) $f(t_1,t_2,t_3)$ is a map from $\mathbb{R}^3$ to $\mathbb{O}$ defined by: 
	\begin{align}
	&\hat{f}_{3}(\omega_1,\omega_2,\omega_3)=	(\mathcal{F}_3^{\mathbb{O}}f)(\omega_1,\omega_2,\omega_3)\nonumber\\&=\int_{-\infty}^{\infty}\int_{-\infty}^{\infty}\int_{-\infty}^{\infty}f(t_1,t_2,t_3)\boldmath{e}^{-e_{1}2\pi (t_1,t_2,t_3)(\omega_1,\omega_2,\omega_3)} \nonumber\\&\times\boldmath{e}^{-e_{2}2\pi (t_1,t_2,t_3)(\omega_1,\omega_2,\omega_3)}\boldmath{e}^{-e_{4}2\pi (t_1,t_2,t_3)(\omega_1,\omega_2,\omega_3)}{\rm d}t_1{\rm d}t_2{\rm d}t_3.
	\end{align}
\end{definition}
\begin{definition}[Inversion formula for 3D OFT] The inverse 3D OFT of the spectrum $(\mathcal{F}_3^{\mathbb{O}}g)$ is given by
	\begin{align*}
	&f(t_1,t_2,t_3)=\left((\mathcal{F}_{3}^{\mathbb{O}}f)^{-1}\hat{f}\right)(t_1,t_2,t_3)\\&=\int_{-\infty}^{\infty}\int_{-\infty}^{\infty}\int_{-\infty}^{\infty}\hat{f}_3(\omega_1,\omega_2,\omega_3)\boldmath{e}^{e_{1}2\pi (t_1,t_2,t_3)(\omega_1,\omega_2,\omega_3)}\\&\times\boldmath{e}^{e_{2}2\pi (t_1,t_2,t_3)(\omega_1,\omega_2,\omega_3)}\boldmath{e}^{e_{4}2\pi (t_1,t_2,t_3)(\omega_1,\omega_2,\omega_3)}{\rm d}\omega_1{\rm d}\omega_2{\rm d}\omega_3.
	\end{align*}
\end{definition}   
\begin{definition}[3D OCLCT]
	The OCLCT of an octonion-valued signal (or function) $f(t_1,t_2,t_3)$ is a map from $\mathbb{R}^3$ to $\mathbb{O}$ defined by: 
	\begin{align}\label{eq:2.2}
	&(\mathcal{L}^{\mathbb{O}}_{\mathcal{N}_{1},\mathcal{N}_{2},\mathcal{N}_{3}}f)(\omega_1,\omega_2,\omega_3)\nonumber\\&=\int_{-\infty}^{\infty}\int_{-\infty}^{\infty}\int_{-\infty}^{\infty}f(t_1,t_2,t_3)\kappa_{\mathcal{N}_{1}}^{e_{1}}(t_{1},\omega_{1})\kappa_{\mathcal{N}_{2}}^{e_{2}}(t_{2},\omega_{2})\kappa_{\mathcal{N}_{3}}^{e_{4}}(t_{3},\omega_{3}){\rm d}t_1{\rm d}t_2{\rm d}t_3,
	\end{align}
	where $\mathcal{N}_{j}=\begin{pmatrix}
	a_{j} & b_{j}\\ 
	c_{j} & d_{j}
	\end{pmatrix}\in{\mathbb{R}^{2\times 2}}$ be a matrix parameter satisfying $det(\mathcal{N}_{j})=1$, for $j=1,2,3$ and the kernel
	\begin{align}\label{eq:2}
	\kappa^{e_{1}}_{\mathcal{N}_{1}}(t_{1},\omega_{1})=
	\left\{ \begin{array}{ll}
	\frac{1}{\sqrt{2\pi|b_{1}|}}\;{\rm e}^{{\rm e_{1}}\left(\frac{a_{1}t_{1}^{2}}{2b_{1}}-\frac{t_{1}\omega_{1}}{b_{1}}
		+\frac{d_{1} \omega_{1}^{2}}{2b_{1}}-\frac{\pi}{2}\right)} &b_{1}\neq 0\\
	\sqrt{d}{\rm e}^{e_{1}\frac{c_{1}d_{1}}{2}\omega_{1}^{2}}\delta(t_{1}-d_{1}\omega_{1}) &b_{1}=0,
	\end{array} \right.
	\end{align}
	\begin{align}\label{eq:3}
	\kappa^{e_{2}}_{\mathcal{N}_{2}}(t_{2},\omega_{2})=
	\left\{ \begin{array}{ll}
	\frac{1}{\sqrt{2\pi|b_{2}|}}\;{\rm e}^{{\rm e_{2}}\left(\frac{a_{2} t_{2}^{2}}{2b_{2}}-\frac{t_{2}\omega_{2}}{b_{2}}
		+\frac{d_{2} \omega_{2}^{2}}{2b_{2}}-\frac{\pi}{2}\right)} &b_{2}\neq 0\\
	\sqrt{d_{2}}{\rm e}^{e_{4}\frac{c_{2}d_{2}}{2}\omega_{2}^{2}}\delta(t_{2}-d_{2}\omega_{2}) &b_{2}=0,
	\end{array} \right.
	\end{align}
	\begin{align}\label{eq:4}
	\kappa^{e_{4}}_{\mathcal{N}_{3}}(t_{3},\omega_{3})=
	\left\{ \begin{array}{ll}
	\frac{1}{\sqrt{2\pi|b_{3}|}}\;{\rm e}^{{\rm e_{4}}\left(\frac{a_{3} t_{3}^{2}}{2b_{3}}-\frac{t_{3}\omega_{3}}{b_{3}}
		+\frac{d_{3} \omega_{3}^{2}}{2b_{3}}-\frac{\pi}{2}\right)} &b_{3}\neq 0\\
	\sqrt{d_{3}}{\rm e}^{e_{4}\frac{c_{3}d_{3}}{2}\omega_{3}^{2}}\delta(t_{3}-d_{3}\omega_{3}) &b_{3}=0,
	\end{array} \right.
	\end{align}
\end{definition}
where $\delta(t_1,t_2,t_3)$ representing the Dirac delta function.
\begin{definition}[Inversion formula for 3D OCLCT]
	The inverse of OCLCT having an octonion-valued signal (or function) $f(t_1,t_2,t_3)$ is a map from $\mathbb{R}^3$ to $\mathbb{O}$ defined by: 
	\begin{align}\label{eq:2.6a}
	&f(t_1,t_2,t_3)\nonumber\\&=\int_{-\infty}^{\infty}\int_{-\infty}^{\infty}\int_{-\infty}^{\infty}(\mathcal{L}_{ \mathcal{N}_1,\mathcal{N}_2,\mathcal{N}_3}^{\mathbb{O}}f)(\omega_1,\omega_2,\omega_3)\kappa_{\mathcal{N}_{1}}^{-e_{1}}(t_{1},\omega_{1})\kappa_{\mathcal{N}_{2}}^{-e_{2}}(t_{2},\omega_{2})\kappa_{\mathcal{N}_{3}}^{-e_{4}}(t_{3},\omega_{3}){\rm d}\omega_1{\rm d}\omega_2{\rm d}\omega_3,
	\end{align}
	where
	\begin{align*}
	&\kappa_{\mathcal{N}_{1}}^{-e_{1}}(t_{1},\omega_{1})=\kappa_{\mathcal{N}_{1}^{-1}}^{e_{1}}(\omega_{1},t_{1})=\frac{1}{\sqrt{2 \pi |b_1|}}e^{-e_1(\frac{a_1t_1^2}{2b_1}-\frac{t_1\omega_1}{b_1}+\frac{d_1\omega_1^2}{2b_1}-\frac{\pi}{2})},\\&\kappa_{\mathcal{N}_{2}}^{-e_{2}}(t_{2},\omega_{2})=\kappa_{\mathcal{N}_{2}^{-1}}^{e_{2}}(\omega_{2},t_{2})=\frac{1}{\sqrt{2 \pi |b_2|}}e^{-e_2(\frac{a_2t_2^2}{2b_2}-\frac{t_2\omega_2}{b_2}+\frac{d_2\omega_2^2}{2b_2}-\frac{\pi}{2})},\\&\kappa_{\mathcal{N}_{3}}^{-e_{4}}(t_{3},\omega_{3})=\kappa_{\mathcal{N}_{3}^{-1}}^{e_{4}}(\omega_{3},t_{3})=\frac{1}{\sqrt{2 \pi |b_3|}}e^{-e_4(\frac{a_3t_3^2}{2b_3}-\frac{t_3\omega_3}{b_3}+\frac{d_3\omega_3^2}{2b_3}-\frac{\pi}{2})},\\&\mathcal{N}_{j}=\begin{pmatrix}
	a_{j} & b_{j}\\ 
	c_{j} & d_{j}
	\end{pmatrix}\in{\mathbb{R}^{2\times 2}}, \mathcal{N}_{j}^{-1}=\begin{pmatrix}
	d_{j} & -b_{j}\\ 
	-c_{j} & a_{j}
	\end{pmatrix}\in{\mathbb{R}^{2\times 2}} \ \ \text{and} \ \ b \neq 0.
	\end{align*}
\end{definition}
\begin{remark}[Relationship between OCLCT and OFT] From \cite{wentheoctonion2021}, one can establish a relationship between OCLCT and OFT with certain parameters of the matrix $\mathcal{N}_j$, for $j=1,2,3$ as follows:
	\begin{align*}
	&(\mathcal{L}^{\mathbb{O}}_{\mathcal{N}_{1},\mathcal{N}_{2},\mathcal{N}_{3}}f)(\omega_1,\omega_2,\omega_3)\\&=\frac{1}{(2\pi)^{\frac{3}{2}}}\int_{-\infty}^{\infty}\int_{-\infty}^{\infty}\int_{-\infty}^{\infty}f(t_1,t_2,t_3)e^{e_{1}\left(-t_{1}\omega_{1}-\frac{\pi}{2}\right)}e^{e_{2}\left(-t_{2}\omega_{2}-\frac{\pi}{2}\right)}e^{e_{4}\left(-t_{3}\omega_{3}-\frac{\pi}{2}\right)}{\rm d}t_1{\rm d}t_2{\rm d}t_3\\
	&=\frac{1}{(2\pi)^{\frac{3}{2}}}\int_{-\infty}^{\infty}\int_{-\infty}^{\infty}\int_{-\infty}^{\infty}f(t_1,t_2,t_3)e^{e_{1}\left(-t_{1}\omega_{1}\right)}(-e_{1})e^{e_{2}  \left(-t_{2}\omega_{2}\right)}\\&\times(-e_{2})e^{e_{4}\left(-t_{3}\omega_{3}\right)}(-e_{4}){\rm d}t_1{\rm d}t_2{\rm d}t_3\\
	&=\frac{1}{(2\pi)^{\frac{3}{2}}}(\mathcal{F}_3^{\mathbb{O}}f)\left(\frac{\omega_{1}}{2\pi},-\frac{\omega_{2}}{2\pi},-\frac{\omega_{3}}{2\pi}\right)e_{7},
	\end{align*}
	where $a_{j}=d_{j}=0,b_{j}=1,c_{j}=-1$.
\end{remark}
\section{Definition and properties of the 3D WOCLCT}\label{sec:3}
Before defining the WOCLCT, we first define the octonion window signal (or function):
\begin{definition}An octonion window (OW) of an octonion-valued signal (or function) $\Psi(t_1,t_2,t_3) \in L^2(\mathbb{R}^3, \mathbb{O})\setminus\{0\}$ defined by:
	\begin{align}\label{eq:3.1}
	\Psi^{(\omega_1,\omega_2,\omega_3)}_{(\mu_1,\mu_2,\mu_3) }(t_1,t_2,t_3)=\kappa_{\mathcal{N}_{1}}^{-e_{1}}(t_{1},\omega_{1})\kappa_{\mathcal{N}_{2}}^{-e_{2}}(t_{2},\omega_{2})\kappa_{\mathcal{N}_{3}}^{-e_{4}}(t_{3},\omega_{3})\Psi(t_1-\mu_1,t_2-\mu_2,t_3-\mu_3)
	\end{align} 
	for each $(t_1,t_2,t_3),(\omega_1,\omega_2,\omega_3)$ and $(\mu_1,\mu_2,\mu_3) \in \mathbb{R}^3$.
\end{definition}
\begin{definition}[3D WOCLCT] The WOCLCT of an octonion-valued signal (or function) $f \in L^2(\mathbb{R}^3,\mathbb{O})$ with respect to OW signal (or function) $\Psi \in L^2(\mathbb{R}^3,\mathbb{O})\setminus\{0\}$ defined by:
	\begin{align}\label{eq:3.2}
	&\{\mathcal{G}_{\mathcal{N}_1,\mathcal{N}_2,\mathcal{N}_3}^{\mathbb{O}}(f,\Psi)\}((\omega_1,\omega_2,\omega_3),(\mu_1,\mu_2,\mu_3))\nonumber\\&=\left<f,	\Psi^{(\omega_1,\omega_2,\omega_3)}_{(\mu_1,\mu_2,\mu_3) }(t_1,t_2,t_3)\right>_{L^2(\mathbb{R}^3,\mathbb{O})}\nonumber \\& =\int_{-\infty}^{\infty}\int_{-\infty}^{\infty}\int_{-\infty}^{\infty}f(t_1,t_2,t_3)\overline{	\Psi^{(\omega_1,\omega_2,\omega_3)}_{(\mu_1,\mu_2,\mu_3) }(t_1,t_2,t_3)}{\rm d}t_1{\rm d}t_2{\rm d}t_3 \nonumber \\&=\int_{-\infty}^{\infty}\int_{-\infty}^{\infty}\int_{-\infty}^{\infty}f(t_1,t_2,t_3)\overline{\Psi(t_1-\mu_1,t_2-\mu_2,t_3-\mu_3)}\kappa_{\mathcal{N}_{1}}^{e_{1}}(t_{1},\omega_{1})\nonumber \\&\times\kappa_{\mathcal{N}_{2}}^{e_{2}}(t_{2},\omega_{2})\kappa_{\mathcal{N}_{3}}^{e_{4}}(t_{3},\omega_{3}){\rm d}t_1{\rm d}t_2{\rm d}t_3.
	\end{align}
\end{definition}
\begin{remark}[Relationship between WOCLCT and OCLCT]
	Let $\Psi$ be an octonion-valued window signal (or function). For every
	$ f \in L^2(\mathbb{R}^3,\mathbb{O})$, then we have
	\begin{align*}
	\{\mathcal{G}_{\mathcal{N}_1,\mathcal{N}_2,\mathcal{N}_3}^{\mathbb{O}}(f,\Psi)\}((\omega_1,\omega_2,\omega_3),(\mu_1,\mu_2,\mu_3))=(\mathcal{L}^{\mathbb{O}}_{\mathcal{N}_{1},\mathcal{N}_{2},\mathcal{N}_{3}}(f\mathcal{T}_{(\mu_1,\mu_2,\mu_3)}\Psi))(\omega_1,\omega_2,\omega_3),
	\end{align*}
	where $\mathcal{T}_{(\mu_1,\mu_2,\mu_3)}\Psi(t_1,t_2,t_3)= \Psi(t_1-\mu_1,t_2-\mu_2,t_3-\mu_3)$.
\end{remark}
\begin{definition}[Inversion formula for 3D WOCLCT] The inverse transform of WOCLCT of an octonion-valued signal (or function) $f \in L^2(\mathbb{R}^3,\mathbb{O})$ with respect to OW signal (or function) $\Psi \in L^2(\mathbb{R}^3,\mathbb{O})\setminus\{0\}$ defined by:
	\begin{align}\label{eq:3.0}
	&f(t_1,t_2,t_3)\nonumber\\&=\frac{1}{||\Psi||^2}\int_{-\infty}^{\infty}\int_{-\infty}^{\infty}\int_{-\infty}^{\infty}\int_{-\infty}^{\infty}\int_{-\infty}^{\infty}\int_{-\infty}^{\infty}\{\mathcal{G}_{\mathcal{N}_1,\mathcal{N}_2,\mathcal{N}_3}^{\mathbb{O}}(f,\Psi)\}((\omega_1,\omega_2,\omega_3),(\mu_1,\mu_2,\mu_3))	\nonumber\\& \times	\Psi^{(\omega_1,\omega_2,\omega_3)}_{(\mu_1,\mu_2,\mu_3)}(t_1,t_2,t_3){\rm d}\omega_1{\rm d}\omega_2{\rm d}\omega_3{\rm d}\mu_1{\rm d}\mu_2{\rm d}\mu_3.
	\end{align}
\end{definition}
\begin{proof}
	To establish the inversion formula for 3D WOCLCT, we use the definition of 3D WOCLCT provided in equation \eqref{eq:3.2} as follows:
	\begin{align}\label{eq:3.3a}
	&\{\mathcal{G}_{\mathcal{N}_1,\mathcal{N}_2,\mathcal{N}_3}^{\mathbb{O}}(f,\Psi)\}((\omega_1,\omega_2,\omega_3),(\mu_1,\mu_2,\mu_3))\nonumber\\&=\int_{-\infty}^{\infty}\int_{-\infty}^{\infty}\int_{-\infty}^{\infty}f(t_1,t_2,t_3)\overline{\Psi(t_1-\mu_1,t_2-\mu_2,t_3-\mu_3)}\kappa_{\mathcal{N}_{1}}^{e_{1}}(t_{1},\omega_{1})\nonumber \\&\times\kappa_{\mathcal{N}_{2}}^{e_{2}}(t_{2},\omega_{2})\kappa_{\mathcal{N}_{3}}^{e_{4}}(t_{3},\omega_{3}){\rm d}t_1{\rm d}t_2{\rm d}t_3
	\end{align}
	Now, applying the inversion formula for the 3D OCLCT provided in equation \eqref{eq:2.6a} on equation \eqref{eq:3.3a}, we have
	\begin{align}\label{eq:3.3}
	&f(t_1,t_2,t_3)\overline{\Psi(t_1-\mu_1,t_2-\mu_2,t_3-\mu_3)}\nonumber\\&=\int_{-\infty}^{\infty}\int_{-\infty}^{\infty}\int_{-\infty}^{\infty}\{\mathcal{G}_{\mathcal{N}_1,\mathcal{N}_2,\mathcal{N}_3}^{\mathbb{O}}(f,\Psi)\}((\omega_1,\omega_2,\omega_3),(\mu_1,\mu_2,\mu_3))\kappa_{\mathcal{N}_{1}}^{-e_{1}}(t_{1},\omega_1)\nonumber\\&\times\kappa_{\mathcal{N}_{2}}^{-e_{2}}(t_{2},\omega_{2})\kappa_{\mathcal{N}_{3}}^{-e_{4}}(t_{3},\omega_{3}){\rm d}\omega_1{\rm d}\omega_2{\rm d}\omega_3
	\end{align}
	Post-multiplying both sides of equation \eqref{eq:3.3} by $\Psi(t_1-\mu_1,t_2-\mu_2,t_3-\mu_3)$ and integrating with respect to $\mu_1,\mu_2$, and $\mu_3$, we get
	\begin{align*}
	&\int_{-\infty}^{\infty}\int_{-\infty}^{\infty}\int_{-\infty}^{\infty}f(t_1,t_2,t_3)\overline{\Psi(t_1-\mu_1,t_2-\mu_2,t_3-\mu_3)}\Psi(t_1-\mu_1,t_2-\mu_2,t_3-\mu_3){\rm d}\mu_1{\rm d}\mu_2{\rm d}\mu_3\\&=\int_{-\infty}^{\infty}\int_{-\infty}^{\infty}\int_{-\infty}^{\infty}\int_{-\infty}^{\infty}\int_{-\infty}^{\infty}\int_{-\infty}^{\infty}\{\mathcal{G}_{\mathcal{N}_1,\mathcal{N}_2,\mathcal{N}_3}^{\mathbb{O}}(f,\Psi)\}((\omega_1,\omega_2,\omega_3),(\mu_1,\mu_2,\mu_3))\kappa_{\mathcal{N}_{1}}^{-e_{1}}(t_{1},\omega_1)\nonumber\\&\times \kappa_{\mathcal{N}_{2}}^{-e_{2}}(t_{2},\omega_{2})\kappa_{\mathcal{N}_{3}}^{-e_{4}}(t_{3},\omega_{3}){\rm d}\omega_1{\rm d}\omega_2{\rm d}\omega_3{\rm d}\mu_1{\rm d}\mu_2{\rm d}\mu_3.
	\end{align*} 
	By using the alternativity property of octonion algebra, we have
	\begin{align*}
	&f(t_1,t_2,t_3)	\int_{-\infty}^{\infty}\int_{-\infty}^{\infty}\int_{-\infty}^{\infty}\left|\Psi(t_1-\mu_1,t_2-\mu_2,t_3-\mu_3)\right|^2{\rm d}\mu_1{\rm d}\mu_2{\rm d}\mu_3\\&=\int_{-\infty}^{\infty}\int_{-\infty}^{\infty}\int_{-\infty}^{\infty}\int_{-\infty}^{\infty}\int_{-\infty}^{\infty}\int_{-\infty}^{\infty}\{\mathcal{G}_{\mathcal{N}_1,\mathcal{N}_2,\mathcal{N}_3}^{\mathbb{O}}(f,\Psi)\}((\omega_1,\omega_2,\omega_3),(\mu_1,\mu_2,\mu_3))\kappa_{\mathcal{N}_{1}}^{-e_{1}}(t_{1},\omega_1)\nonumber\\&\times\kappa_{\mathcal{N}_{2}}^{-e_{2}}(t_{2},\omega_{2})\kappa_{\mathcal{N}_{3}}^{-e_{4}}(t_{3},\omega_{3}){\rm d}\omega_1{\rm d}\omega_2{\rm d}\omega_3{\rm d}\mu_1{\rm d}\mu_2{\rm d}\mu_3.
	\end{align*}
	Using equation \eqref{eq:3.1} on the right-hand side of the above equation, we get
	\begin{align}\label{eq:3.3ab}
	&f(t_1,t_2,t_3)\nonumber\\&=\frac{1}{||\Psi||^2}\int_{-\infty}^{\infty}\int_{-\infty}^{\infty}\int_{-\infty}^{\infty}\int_{-\infty}^{\infty}\int_{-\infty}^{\infty}\int_{-\infty}^{\infty}\{\mathcal{G}_{\mathcal{N}_1,\mathcal{N}_2,\mathcal{N}_3}^{\mathbb{O}}(f,\Psi)\}((\omega_1,\omega_2,\omega_3),(\mu_1,\mu_2,\mu_3)) 	\nonumber\\&\times\Psi^{(\omega_1,\omega_2,\omega_3)}_{(\mu_1,\mu_2,\mu_3)}(t_1,t_2,t_3){\rm d}\omega_1{\rm d}\omega_2{\rm d}\omega_3{\rm d}\mu_1{\rm d}\mu_2{\rm d}\mu_3.
	\end{align}
	Thus, the desired result is obtained.
\end{proof}
Now, the goal is to define various important properties of the 3D WOCLCT. To achieve this goal, we first expand the kernel of 3D WOCLCT given in equation \eqref{eq:3.2} in full octonion form as follows:
\begin{align}\label{eq:2.6}
&\left(\kappa^{e_{1}}_{\mathcal{N}_{1}}(t_{1},\omega_{1})\kappa^{e_{2}}_{\mathcal{N}_{2}}(t_{2},\omega_{2})\right)\kappa^{e_{4}}_{\mathcal{N}_{3}}(t_{3},\omega_{3})\nonumber\\
&=\frac{1}{(2\pi)^{\frac{3}{2}}\sqrt{|b_{1}b_{2}b_{3}|}}\left[e^{e_1\left(\frac{a_1t_1^2}{2b_1}-\frac{t_1\omega_1}{b_1}+\frac{d_1\omega_1^2}{2b_1}-\frac{\pi}{2}\right)}e^{e_2\left(\frac{a_2t_2^2}{2b_2}-\frac{t_2\omega_2}{b_2}+\frac{d_2\omega_2^2}{2b_2}-\frac{\pi}{2}\right)}\right]e^{e_4\left(\frac{a_3t_3^2}{2b_3}-\frac{t_3\omega_3}{b_3}+\frac{d_3\omega_3^2}{2b_3}-\frac{\pi}{2}\right)}\nonumber\\
&=\frac{1}{(2\pi)^{\frac{3}{2}}\sqrt{|b_{1}b_{2}b_{3}|}}\Bigg[\cos{\left(\frac{a_1t_1^2}{2b_1}-\frac{t_1\omega_1}{b_1}+\frac{d_1\omega_1^2}{2b_1}-\frac{\pi}{2}\right)}\cos{\left(\frac{a_2t_2^2}{2b_2}-\frac{t_2\omega_2}{b_2}+\frac{d_2\omega_2^2}{2b_2}-\frac{\pi}{2}\right)}\nonumber\\&\times\cos{\left(\frac{a_3t_3^2}{2b_3}-\frac{t_3\omega_3}{b_3}+\frac{d_3\omega_3^2}{2b_3}-\frac{\pi}{2}\right)}+\sin{\left(\frac{a_1t_1^2}{2b_1}-\frac{t_1\omega_1}{b_1}+\frac{d_1\omega_1^2}{2b_1}-\frac{\pi}{2}\right)}\cos{\left(\frac{a_2t_2^2}{2b_2}-\frac{t_2\omega_2}{b_2}+\frac{d_2\omega_2^2}{2b_2}-\frac{\pi}{2}\right)}\nonumber\\&\times\cos{\left(\frac{a_3t_3^2}{2b_3}-\frac{t_3\omega_3}{b_3}+\frac{d_3\omega_3^2}{2b_3}-\frac{\pi}{2}\right)}e_{1}+\cos{\left(\frac{a_1t_1^2}{2b_1}-\frac{t_1\omega_1}{b_1}+\frac{d_1\omega_1^2}{2b_1}-\frac{\pi}{2}\right)}\sin{\left(\frac{a_2t_2^2}{2b_2}-\frac{t_2\omega_2}{b_2}+\frac{d_2\omega_2^2}{2b_2}-\frac{\pi}{2}\right)}\nonumber\\&\times\cos{\left(\frac{a_3t_3^2}{2b_3}-\frac{t_3\omega_3}{b_3}+\frac{d_3\omega_3^2}{2b_3}-\frac{\pi}{2}\right)}e_2+\sin{\left(\frac{a_1t_1^2}{2b_1}-\frac{t_1\omega_1}{b_1}+\frac{d_1\omega_1^2}{2b_1}-\frac{\pi}{2}\right)}\sin{\left(\frac{a_2t_2^2}{2b_2}-\frac{t_2\omega_2}{b_2}+\frac{d_2\omega_2^2}{2b_2}-\frac{\pi}{2}\right)}\nonumber\\&\times\cos{\left(\frac{a_3t_3^2}{2b_3}-\frac{t_3\omega_3}{b_3}+\frac{d_3\omega_3^2}{2b_3}-\frac{\pi}{2}\right)}e_{3}+\cos{\left(\frac{a_1t_1^2}{2b_1}-\frac{t_1\omega_1}{b_1}+\frac{d_1\omega_1^2}{2b_1}-\frac{\pi}{2}\right)}\cos{\left(\frac{a_2t_2^2}{2b_2}-\frac{t_2\omega_2}{b_2}+\frac{d_2\omega_2^2}{2b_2}-\frac{\pi}{2}\right)}\nonumber\\&\times\sin{\left(\frac{a_3t_3^2}{2b_3}-\frac{t_3\omega_3}{b_3}+\frac{d_3\omega_3^2}{2b_3}-\frac{\pi}{2}\right)}e_4+\sin{\left(\frac{a_1t_1^2}{2b_1}-\frac{t_1\omega_1}{b_1}+\frac{d_1\omega_1^2}{2b_1}-\frac{\pi}{2}\right)}\cos{\left(\frac{a_2t_2^2}{2b_2}-\frac{t_2\omega_2}{b_2}+\frac{d_2\omega_2^2}{2b_2}-\frac{\pi}{2}\right)}\nonumber\\&\times\sin{\left(\frac{a_3t_3^2}{2b_3}-\frac{t_3\omega_3}{b_3}+\frac{d_3\omega_3^2}{2b_3}-\frac{\pi}{2}\right)}e_{5}+\cos{\left(\frac{a_1t_1^2}{2b_1}-\frac{t_1\omega_1}{b_1}+\frac{d_1\omega_1^2}{2b_1}-\frac{\pi}{2}\right)}\sin{\left(\frac{a_2t_2^2}{2b_2}-\frac{t_2\omega_2}{b_2}+\frac{d_2\omega_2^2}{2b_2}-\frac{\pi}{2}\right)}\nonumber\\&\times\sin{\left(\frac{a_3t_3^2}{2b_3}-\frac{t_3\omega_3}{b_3}+\frac{d_3\omega_3^2}{2b_3}-\frac{\pi}{2}\right)}e_{6}+\sin{\left(\frac{a_1t_1^2}{2b_1}-\frac{t_1\omega_1}{b_1}+\frac{d_1\omega_1^2}{2b_1}-\frac{\pi}{2}\right)}\sin{\left(\frac{a_2t_2^2}{2b_2}-\frac{t_2\omega_2}{b_2}+\frac{d_2\omega_2^2}{2b_2}-\frac{\pi}{2}\right)}\nonumber\\&\times\sin{\left(\frac{a_3t_3^2}{2b_3}-\frac{t_3\omega_3}{b_3}+\frac{d_3\omega_3^2}{2b_3}-\frac{\pi}{2}\right)}e_{7}\Bigg].
\end{align}
On the basis of full octonion form \eqref{eq:2.6} of kernel, the 3D OW signal (or function) $f(t_1,t_2,t_3)$ is a map from $\mathbb{R}^3$ to $\mathbb{O}$ defined by:
\begin{align}\label{eq:3.66}
\{\mathcal{G}_{\mathcal{N}_1,\mathcal{N}_2,\mathcal{N}_3}^{\mathbb{O}}(f,\Psi)\}&=\{\mathcal{G}_{\mathcal{N}_1,\mathcal{N}_2,\mathcal{N}_3}^{\mathbb{O}}(f,\Psi)\}_{eee}+\{\mathcal{G}_{\mathcal{N}_1,\mathcal{N}_2,\mathcal{N}_3}^{\mathbb{O}}(f,\Psi)\}_{oee}e_1+\{\mathcal{G}_{\mathcal{N}_1,\mathcal{N}_2,\mathcal{N}_3}^{\mathbb{O}}(f,\Psi)\}_{eoe}e_2\nonumber\\&+\{\mathcal{G}_{\mathcal{N}_1,\mathcal{N}_2,\mathcal{N}_3}^{\mathbb{O}}(f,\Psi)\}_{ooe}e_3+\{\mathcal{G}_{\mathcal{N}_1,\mathcal{N}_2,\mathcal{N}_3}^{\mathbb{O}}(f,\Psi)\}_{eeo}e_4+\{\mathcal{G}_{\mathcal{N}_1,\mathcal{N}_2,\mathcal{N}_3}^{\mathbb{O}}(f,\Psi)\}_{oeo}e_5\nonumber \\&+\{\mathcal{G}_{\mathcal{N}_1,\mathcal{N}_2,\mathcal{N}_3}^{\mathbb{O}}(f,\Psi)\}_{eoo}e_6+\{\mathcal{G}_{\mathcal{N}_1,\mathcal{N}_2,\mathcal{N}_3}^{\mathbb{O}}(f,\Psi)\}_{ooo}e_7,
\end{align}
where
\begin{align*}
&\{\mathcal{G}_{\mathcal{N}_1,\mathcal{N}_2,\mathcal{N}_3}^{\mathbb{O}}(f,\Psi)\}_{eee}(\omega_1,\omega_2,\omega_3)\\&=\frac{1}{(2\pi)^{\frac{3}{2}}\sqrt{|b_{1}b_{2}b_{3}|}}\int_{-\infty}^{\infty}\int_{-\infty}^{\infty}\int_{-\infty}^{\infty}f_{eee}(t_1,t_2,t_3)\overline{\Psi_{eee}(t_1-\mu_1,t_2-\mu_2,t_3-\mu_3)}\\&\times\cos{\left(\frac{a_1t_1^2}{2b_1}-\frac{t_1\omega_1}{b_1}+\frac{d_1\omega_1^2}{2b_1}-\frac{\pi}{2}\right)}\cos{\left(\frac{a_2t_2^2}{2b_2}-\frac{t_2\omega_2}{b_2}+\frac{d_2\omega_2^2}{2b_2}-\frac{\pi}{2}\right)}\nonumber\\&\times\cos{\left(\frac{a_3t_3^2}{2b_3}-\frac{t_3\omega_3}{b_3}+\frac{d_3\omega_3^2}{2b_3}-\frac{\pi}{2}\right)}{\rm d}t_1{\rm d}t_2{\rm d}t_3,\\&\{\mathcal{G}_{\mathcal{N}_1,\mathcal{N}_2,\mathcal{N}_3}^{\mathbb{O}}(f,\Psi)\}_{oee}(\omega_1,\omega_2,\omega_3)\\&=\frac{1}{(2\pi)^{\frac{3}{2}}\sqrt{|b_{1}b_{2}b_{3}|}}\int_{-\infty}^{\infty}\int_{-\infty}^{\infty}\int_{-\infty}^{\infty}f_{oee}(t_1,t_2,t_3)\overline{\Psi_{oee}(t_1-\mu_1,t_2-\mu_2,t_3-\mu_3)}\\&\times\sin{\left(\frac{a_1t_1^2}{2b_1}-\frac{t_1\omega_1}{b_1}+\frac{d_1\omega_1^2}{2b_1}-\frac{\pi}{2}\right)}\cos{\left(\frac{a_2t_2^2}{2b_2}-\frac{t_2\omega_2}{b_2}+\frac{d_2\omega_2^2}{2b_2}-\frac{\pi}{2}\right)}\nonumber\\&\times\cos{\left(\frac{a_3t_3^2}{2b_3}-\frac{t_3\omega_3}{b_3}+\frac{d_3\omega_3^2}{2b_3}-\frac{\pi}{2}\right)}{\rm d}t_1{\rm d}t_2{\rm d}t_3,\\&\{\mathcal{G}_{\mathcal{N}_1,\mathcal{N}_2,\mathcal{N}_3}^{\mathbb{O}}(f,\Psi)\}_{eoe}(\omega_1,\omega_2,\omega_3)\\&=\frac{1}{(2\pi)^{\frac{3}{2}}\sqrt{|b_{1}b_{2}b_{3}|}}\int_{-\infty}^{\infty}\int_{-\infty}^{\infty}\int_{-\infty}^{\infty}f_{eoe}(t_1,t_2,t_3)\overline{\Psi_{eoe}(t_1-\mu_1,t_2-\mu_2,t_3-\mu_3)}\\&\times\cos{\left(\frac{a_1t_1^2}{2b_1}-\frac{t_1\omega_1}{b_1}+\frac{d_1\omega_1^2}{2b_1}-\frac{\pi}{2}\right)}\sin{\left(\frac{a_2t_2^2}{2b_2}-\frac{t_2\omega_2}{b_2}+\frac{d_2\omega_2^2}{2b_2}-\frac{\pi}{2}\right)}\nonumber\\&\times\cos{\left(\frac{a_3t_3^2}{2b_3}-\frac{t_3\omega_3}{b_3}+\frac{d_3\omega_3^2}{2b_3}-\frac{\pi}{2}\right)}{\rm d}t_1{\rm d}t_2{\rm d}t_3,\\&\{\mathcal{G}_{\mathcal{N}_1,\mathcal{N}_2,\mathcal{N}_3}^{\mathbb{O}}(f,\Psi)\}_{ooe}(\omega_1,\omega_2,\omega_3)\\&=\frac{1}{(2\pi)^{\frac{3}{2}}\sqrt{|b_{1}b_{2}b_{3}|}}\int_{-\infty}^{\infty}\int_{-\infty}^{\infty}\int_{-\infty}^{\infty}f_{ooe}(t_1,t_2,t_3)\overline{\Psi_{ooe}(t_1-\mu_1,t_2-\mu_2,t_3-\mu_3)}\\&\times\sin{\left(\frac{a_1t_1^2}{2b_1}-\frac{t_1\omega_1}{b_1}+\frac{d_1\omega_1^2}{2b_1}-\frac{\pi}{2}\right)}\sin{\left(\frac{a_2t_2^2}{2b_2}-\frac{t_2\omega_2}{b_2}+\frac{d_2\omega_2^2}{2b_2}-\frac{\pi}{2}\right)}\nonumber\\&\times\cos{\left(\frac{a_3t_3^2}{2b_3}-\frac{t_3\omega_3}{b_3}+\frac{d_3\omega_3^2}{2b_3}-\frac{\pi}{2}\right)}{\rm d}t_1{\rm d}t_2{\rm d}t_3,
\end{align*}
\begin{align*}
&\{\mathcal{G}_{\mathcal{N}_1,\mathcal{N}_2,\mathcal{N}_3}^{\mathbb{O}}(f,\Psi)\}_{eeo}(\omega_1,\omega_2,\omega_3)\\&=\frac{1}{(2\pi)^{\frac{3}{2}}\sqrt{|b_{1}b_{2}b_{3}|}}\int_{-\infty}^{\infty}\int_{-\infty}^{\infty}\int_{-\infty}^{\infty}f_{eeo}(t_1,t_2,t_3)\overline{\Psi_{eeo}(t_1-\mu_1,t_2-\mu_2,t_3-\mu_3)}\\&\times\cos{\left(\frac{a_1t_1^2}{2b_1}-\frac{t_1\omega_1}{b_1}+\frac{d_1\omega_1^2}{2b_1}-\frac{\pi}{2}\right)}\cos{\left(\frac{a_2t_2^2}{2b_2}-\frac{t_2\omega_2}{b_2}+\frac{d_2\omega_2^2}{2b_2}-\frac{\pi}{2}\right)}\nonumber\\&\times\sin{\left(\frac{a_3t_3^2}{2b_3}-\frac{t_3\omega_3}{b_3}+\frac{d_3\omega_3^2}{2b_3}-\frac{\pi}{2}\right)}{\rm d}t_1{\rm d}t_2{\rm d}t_3,\\&\{\mathcal{G}_{\mathcal{N}_1,\mathcal{N}_2,\mathcal{N}_3}^{\mathbb{O}}(f,\Psi)\}_{oeo}(\omega_1,\omega_2,\omega_3)\\&=\frac{1}{(2\pi)^{\frac{3}{2}}\sqrt{|b_{1}b_{2}b_{3}|}}\int_{-\infty}^{\infty}\int_{-\infty}^{\infty}\int_{-\infty}^{\infty}f_{oeo}(t_1,t_2,t_3)\overline{\Psi_{oeo}(t_1-\mu_1,t_2-\mu_2,t_3-\mu_3)}\\&\times\sin{\left(\frac{a_1t_1^2}{2b_1}-\frac{t_1\omega_1}{b_1}+\frac{d_1\omega_1^2}{2b_1}-\frac{\pi}{2}\right)}\cos{\left(\frac{a_2t_2^2}{2b_2}-\frac{t_2\omega_2}{b_2}+\frac{d_2\omega_2^2}{2b_2}-\frac{\pi}{2}\right)}\nonumber\\&\times\sin{\left(\frac{a_3t_3^2}{2b_3}-\frac{t_3\omega_3}{b_3}+\frac{d_3\omega_3^2}{2b_3}-\frac{\pi}{2}\right)}{\rm d}t_1{\rm d}t_2{\rm d}t_3,\end{align*} \begin{align*}&\{\mathcal{G}_{\mathcal{N}_1,\mathcal{N}_2,\mathcal{N}_3}^{\mathbb{O}}(f,\Psi)\}_{eoo}(\omega_1,\omega_2,\omega_3)\\&=\frac{1}{(2\pi)^{\frac{3}{2}}\sqrt{|b_{1}b_{2}b_{3}|}}\int_{-\infty}^{\infty}\int_{-\infty}^{\infty}\int_{-\infty}^{\infty}f_{eoo}(t_1,t_2,t_3)\overline{\Psi_{eoo}(t_1-\mu_1,t_2-\mu_2,t_3-\mu_3)}\\&\times\cos{\left(\frac{a_1t_1^2}{2b_1}-\frac{t_1\omega_1}{b_1}+\frac{d_1\omega_1^2}{2b_1}-\frac{\pi}{2}\right)}\sin{\left(\frac{a_2t_2^2}{2b_2}-\frac{t_2\omega_2}{b_2}+\frac{d_2\omega_2^2}{2b_2}-\frac{\pi}{2}\right)}\nonumber\\&\times\sin{\left(\frac{a_3t_3^2}{2b_3}-\frac{t_3\omega_3}{b_3}+\frac{d_3\omega_3^2}{2b_3}-\frac{\pi}{2}\right)}{\rm d}t_1{\rm d}t_2{\rm d}t_3,
\end{align*}
and
\begin{align*}
&\{\mathcal{G}_{\mathcal{N}_1,\mathcal{N}_2,\mathcal{N}_3}^{\mathbb{O}}(f,\Psi)\}_{ooo}(\omega_1,\omega_2,\omega_3)\\&=\frac{1}{(2\pi)^{\frac{3}{2}}\sqrt{|b_{1}b_{2}b_{3}|}}\int_{-\infty}^{\infty}\int_{-\infty}^{\infty}\int_{-\infty}^{\infty}f_{ooo}(t_1,t_2,t_3)\overline{\Psi_{ooo}(t_1-\mu_1,t_2-\mu_2,t_3-\mu_3)}\\&\times\sin{\left(\frac{a_1t_1^2}{2b_1}-\frac{t_1\omega_1}{b_1}+\frac{d_1\omega_1^2}{2b_1}-\frac{\pi}{2}\right)}\sin{\left(\frac{a_2t_2^2}{2b_2}-\frac{t_2\omega_2}{b_2}+\frac{d_2\omega_2^2}{2b_2}-\frac{\pi}{2}\right)}\nonumber\\&\times\sin{\left(\frac{a_3t_3^2}{2b_3}-\frac{t_3\omega_3}{b_3}+\frac{d_3\omega_3^2}{2b_3}-\frac{\pi}{2}\right)}{\rm d}t_1{\rm d}t_2{\rm d}t_3.
\end{align*}
\begin{proposition}[Linearity property for 3D WOCLCT] The WOCLCT of an octonion-valued signal (or function) $f,g \in L^2(\mathbb{R}^3,\mathbb{O})$ with respect to OW signal (or function) $\Psi \in L^2(\mathbb{R}^3,\mathbb{O})\setminus\{0\}$, then
	\begin{align*}
	&\{\mathcal{G}_{\mathcal{N}_1,\mathcal{N}_2,\mathcal{N}_3}(\eta f+\lambda g,\Psi)\}((\omega_1,\omega_2,\omega_3),(\mu_1,\mu_2,\mu_3))\\&=\eta\{\mathcal{G}_{\mathcal{N}_1,\mathcal{N}_2,\mathcal{N}_3}^{\mathbb{O}}(f,\Psi)\}((\omega_1,\omega_2,\omega_3),(\mu_1,\mu_2,\mu_3))+\lambda \{\mathcal{G}_{\mathcal{N}_1,\mathcal{N}_2,\mathcal{N}_3}(g,\Psi)\}((\omega_1,\omega_2,\omega_3),(\mu_1,\mu_2,\mu_3)),
	\end{align*}
	where $\eta$ and $\lambda$ are any arbitrary octonion constants.
\end{proposition}
\begin{proposition}[Parity for 3D WOCLCT]The WOCLCT of an octonion-valued signal (or function) $f\in L^2(\mathbb{R}^3,\mathbb{O})$ with respect to OW signal (or function) $\Psi \in L^2(\mathbb{R}^3,\mathbb{O})\setminus\{0\}$, then
	\begin{align*}
	\{\mathcal{G}_{\mathcal{N}_1,\mathcal{N}_2,\mathcal{N}_3}(Pf,P\Psi)\}((\omega_1,\omega_2,\omega_3),(\mu_1,\mu_2,\mu_3))=-\{\mathcal{G}_{\mathcal{N}_1,\mathcal{N}_2,\mathcal{N}_3}^{\mathbb{O}}(f,\Psi)\}(-\omega,-\mu),
	\end{align*}
	where $Pf(t_1,t_2,t_3)= f(-t_1,-t_2,-t_3)$.
\end{proposition}
\begin{proof}
	For every $f\in L^2(\mathbb{R}^3,\mathbb{O})$, by direct calculation, we get
	\begin{align}\label{eq:3.7}
	&\{\mathcal{G}_{\mathcal{N}_1,\mathcal{N}_2,\mathcal{N}_3}(Pf,P\Psi)\}((\omega_1,\omega_2,\omega_3),(\mu_1,\mu_2,\mu_3))\nonumber\\&=\int_{-\infty}^{\infty}\int_{-\infty}^{\infty}\int_{-\infty}^{\infty}f(-t_1,-t_2,-t_3)\overline{\Psi(-(t_1-\mu_1,t_2-\mu_2,t_3-\mu_3))}\kappa_{\mathcal{N}_{1}}^{e_{1}}(t_{1},\omega_{1})\nonumber \\&\times\kappa_{\mathcal{N}_{2}}^{e_{2}}(t_{2},\omega_{2})\kappa_{\mathcal{N}_{3}}^{e_{4}}(t_{3},\omega_{3}){\rm d}t_1{\rm d}t_2{\rm d}t_3.		
	\end{align}
	On setting $-(t_1,t_2,t_3)=(x_1,x_2,x_3)$, then we can write the right-hand side of equation \eqref{eq:3.7} as follows:
	\begin{align*}
	&=-\int_{-\infty}^{\infty}\int_{-\infty}^{\infty}\int_{-\infty}^{\infty}f(x_1,x_2,x_3)\overline{\Psi(x_1+\mu_1,x_2+\mu_2,x_3+\mu_3))}\kappa_{\mathcal{N}_{1}}^{e_{1}}(t_{1},-\omega_{1})\nonumber \\&\times\kappa_{\mathcal{N}_{2}}^{e_{2}}(t_{2},-\omega_{2})\kappa_{\mathcal{N}_{3}}^{e_{4}}(t_{3},-\omega_{3}){\rm d}t_1{\rm d}t_2{\rm d}t_3\\&=-\{\mathcal{G}_{\mathcal{N}_1,\mathcal{N}_2,\mathcal{N}_3}^{\mathbb{O}}(f,\Psi)\}(-\omega,-\mu).
	\end{align*}
	Hence, we get the desired result.
\end{proof}
\begin{proposition}[Shifting property for 3D WOCLCT] Let $\mathcal{G}_{\mathcal{N}_1,\mathcal{N}_2,\mathcal{N}_3}^{\mathbb{O}}(f,\Psi)$ be the WOCLCT of the 3D octonion-valued signal (or function) $f$. Suppose $\mathcal{G}_{\mathcal{N}_1,\mathcal{N}_2,\mathcal{N}_3}^{\mathbb{O},s_1}(f,\Psi)\\$,$\mathcal{G}_{\mathcal{N}_1,\mathcal{N}_2,\mathcal{N}_3}^{\mathbb{O},s_2}(f,\Psi)$, and       $\mathcal{G}_{\mathcal{N}_1,\mathcal{N}_2,\mathcal{N}_3}^{\mathbb{O},s_3}(f,\Psi)$ denote the WOCLCT of $f(t_1-s_1,t_2,t_3),f(t_1,t_2-s_2,t_3)$, and$f(t_1,t_2,t_3-s_3)$ respectively, then
	\begin{align*}
	\{\mathcal{G}_{\mathcal{N}_1,\mathcal{N}_2,\mathcal{N}_3}^{\mathbb{O},s_1}(f,\Psi)\}(\omega_{1},\omega_{2},\omega_{3})&=\cos(s_1\omega_1c_1 - \frac{a_1c_1s_1^2}{2})	\{\mathcal{G}_{\mathcal{N}_1,\mathcal{N}_2,\mathcal{N}_3}^{\mathbb{O}}(f,\Psi)\}(\rho_1,\omega_2,\omega_3)\\&-\sin(s_1\omega_1c_1 - \frac{a_1c_1s_1^2}{2})\Delta_1 f(\rho_1,\omega_2,\omega_3),
	\end{align*}
	\begin{align*}
	\{\mathcal{G}_{\mathcal{N}_1,\mathcal{N}_2,\mathcal{N}_3}^{\mathbb{O},s_2}(f,\Psi)\}(\omega_{1},\omega_{2},\omega_{3})&=\cos(s_2\omega_2c_2 - \frac{a_2c_2s_2^2}{2})	\{\mathcal{G}_{\mathcal{N}_1,\mathcal{N}_2,\mathcal{N}_3}^{\mathbb{O}}(f,\Psi)\}(\omega_1,\rho_2,\omega_3)\\&-\sin(s_2\omega_2c_2 - \frac{a_2c_2s_2^2}{2})\Delta_2 f(\omega_1,\rho_2,\omega_3),
	\end{align*}
	and
	\begin{align*}
	\{\mathcal{G}_{\mathcal{N}_1,\mathcal{N}_2,\mathcal{N}_3}^{\mathbb{O},s_3}(f,\Psi)\}(\omega_{1},\omega_{2},\omega_{3})&=\cos(s_3\omega_3c_3 - \frac{a_3c_3s_3^2}{2})	\{\mathcal{G}_{\mathcal{N}_1,\mathcal{N}_2,\mathcal{N}_3}^{\mathbb{O}}(f,\Psi)\}(\omega_1,\omega_2,\rho_3)\\&-\sin(s_3\omega_3c_3 - \frac{a_3c_3s_3^2}{2})\Delta_3 f(\omega_1,\omega_2,\rho_3),
	\end{align*}
	where $\rho_i = \mu_i-s_i$ for $i=1,2$, and $3$, we have
	\begin{align*}
	\Delta_1 f&=	\{\mathcal{G}_{\mathcal{N}_1,\mathcal{N}_2,\mathcal{N}_3}^{\mathbb{O}}(f,\Psi)\}_{see}-\{\mathcal{G}_{\mathcal{N}_1,\mathcal{N}_2,\mathcal{N}_3}^{\mathbb{O}}(f,\Psi)\}_{cee}e_1+\{\mathcal{G}_{\mathcal{N}_1,\mathcal{N}_2,\mathcal{N}_3}^{\mathbb{O}}(f,\Psi)\}_{soe}e_2\nonumber\\&-\{\mathcal{G}_{\mathcal{N}_1,\mathcal{N}_2,\mathcal{N}_3}^{\mathbb{O}}(f,\Psi)\}_{coe}e_3+\{\mathcal{G}_{\mathcal{N}_1,\mathcal{N}_2,\mathcal{N}_3}^{\mathbb{O}}(f,\Psi)\}_{seo}e_4-\{\mathcal{G}_{\mathcal{N}_1,\mathcal{N}_2,\mathcal{N}_3}^{\mathbb{O}}(f,\Psi)\}_{ceo}e_5\nonumber \\&+\{\mathcal{G}_{\mathcal{N}_1,\mathcal{N}_2,\mathcal{N}_3}^{\mathbb{O}}(f,\Psi)\}_{soo}e_6-\{\mathcal{G}_{\mathcal{N}_1,\mathcal{N}_2,\mathcal{N}_3}^{\mathbb{O}}(f,\Psi)\}_{coo}e_7,
	\end{align*}
	\begin{align*}
	\Delta_2 f&=	\{\mathcal{G}_{\mathcal{N}_1,\mathcal{N}_2,\mathcal{N}_3}^{\mathbb{O}}(f,\Psi)\}_{ese}+\{\mathcal{G}_{\mathcal{N}_1,\mathcal{N}_2,\mathcal{N}_3}^{\mathbb{O}}(f,\Psi)\}_{ose}e_1-\{\mathcal{G}_{\mathcal{N}_1,\mathcal{N}_2,\mathcal{N}_3}^{\mathbb{O}}(f,\Psi)\}_{ece}e_2\nonumber\\&-\{\mathcal{G}_{\mathcal{N}_1,\mathcal{N}_2,\mathcal{N}_3}^{\mathbb{O}}(f,\Psi)\}_{oce}e_3+\{\mathcal{G}_{\mathcal{N}_1,\mathcal{N}_2,\mathcal{N}_3}^{\mathbb{O}}(f,\Psi)\}_{eso}e_4+\{\mathcal{G}_{\mathcal{N}_1,\mathcal{N}_2,\mathcal{N}_3}^{\mathbb{O}}(f,\Psi)\}_{oso}e_5\nonumber \\&-\{\mathcal{G}_{\mathcal{N}_1,\mathcal{N}_2,\mathcal{N}_3}^{\mathbb{O}}(f,\Psi)\}_{eco}e_6-\{\mathcal{G}_{\mathcal{N}_1,\mathcal{N}_2,\mathcal{N}_3}^{\mathbb{O}}(f,\Psi)\}_{oco}e_7,
	\end{align*}
	and
	\begin{align*}
	\Delta_1 f&=	\{\mathcal{G}_{\mathcal{N}_1,\mathcal{N}_2,\mathcal{N}_3}^{\mathbb{O}}(f,\Psi)\}_{ees}+\{\mathcal{G}_{\mathcal{N}_1,\mathcal{N}_2,\mathcal{N}_3}^{\mathbb{O}}(f,\Psi)\}_{oes}e_1+\{\mathcal{G}_{\mathcal{N}_1,\mathcal{N}_2,\mathcal{N}_3}^{\mathbb{O}}(f,\Psi)\}_{eos}e_2\nonumber\\&+\{\mathcal{G}_{\mathcal{N}_1,\mathcal{N}_2,\mathcal{N}_3}^{\mathbb{O}}(f,\Psi)\}_{oos}e_3-\{\mathcal{G}_{\mathcal{N}_1,\mathcal{N}_2,\mathcal{N}_3}^{\mathbb{O}}(f,\Psi)\}_{eec}e_4-\{\mathcal{G}_{\mathcal{N}_1,\mathcal{N}_2,\mathcal{N}_3}^{\mathbb{O}}(f,\Psi)\}_{oec}e_5\nonumber \\&-\{\mathcal{G}_{\mathcal{N}_1,\mathcal{N}_2,\mathcal{N}_3}^{\mathbb{O}}(f,\Psi)\}_{eoc}e_6-\{\mathcal{G}_{\mathcal{N}_1,\mathcal{N}_2,\mathcal{N}_3}^{\mathbb{O}}(f,\Psi)\}_{ooc}e_7.
	\end{align*}
\end{proposition}
\begin{proof}
	We can rewrite the equation \eqref{eq:3.66} for the function $f^{s_1}$ as follows:
	\begin{align*}
	&\{\mathcal{G}_{\mathcal{N}_1,\mathcal{N}_2,\mathcal{N}_3}^{\mathbb{O},s_1}(f,\Psi)\}_{eee}(\omega_1,\omega_2,\omega_3)\\&=\frac{1}{(2\pi)^{\frac{3}{2}}\sqrt{|b_{1}b_{2}b_{3}|}}\int_{-\infty}^{\infty}\int_{-\infty}^{\infty}\int_{-\infty}^{\infty}f_{eee}^{s_1}(t_1,t_2,t_3)\overline{\Psi_{eee}(t_1-\mu_1,t_2-\mu_2,t_3-\mu_3)}\\&\times\cos{\left(\frac{a_1t_1^2}{2b_1}-\frac{t_1\omega_1}{b_1}+\frac{d_1\omega_1^2}{2b_1}-\frac{\pi}{2}\right)}\cos{\left(\frac{a_2t_2^2}{2b_2}-\frac{t_2\omega_2}{b_2}+\frac{d_2\omega_2^2}{2b_2}-\frac{\pi}{2}\right)}\nonumber\\&\times\cos{\left(\frac{a_3t_3^2}{2b_3}-\frac{t_3\omega_3}{b_3}+\frac{d_3\omega_3^2}{2b_3}-\frac{\pi}{2}\right)}{\rm d}t_1{\rm d}t_2{\rm d}t_3\\&=\frac{1}{(2\pi)^{\frac{3}{2}}\sqrt{|b_{1}b_{2}b_{3}|}}\int_{-\infty}^{\infty}\int_{-\infty}^{\infty}\int_{-\infty}^{\infty}f_{eee}(t_1-s_1,t_2,t_3)\overline{\Psi_{eee}(t_1-\mu_1,t_2-\mu_2,t_3-\mu_3)}\\&\times\cos{\left(\frac{a_1t_1^2}{2b_1}-\frac{t_1\omega_1}{b_1}+\frac{d_1\omega_1^2}{2b_1}-\frac{\pi}{2}\right)}\cos{\left(\frac{a_2t_2^2}{2b_2}-\frac{t_2\omega_2}{b_2}+\frac{d_2\omega_2^2}{2b_2}-\frac{\pi}{2}\right)}\nonumber\\&\times\cos{\left(\frac{a_3t_3^2}{2b_3}-\frac{t_3\omega_3}{b_3}+\frac{d_3\omega_3^2}{2b_3}-\frac{\pi}{2}\right)}{\rm d}t_1{\rm d}t_2{\rm d}t_3.
	\end{align*}
	Now, using the change of variable technique $x_1=t_1-s_1,x_2=t_2,x_3=t_3$, where $(x_1,x_2,x_3)\in \mathbb{R}^3$, then we obtain
	\begin{align*}
	&\{\mathcal{G}_{\mathcal{N}_1,\mathcal{N}_2,\mathcal{N}_3}^{\mathbb{O},s_1}(f,\Psi)\}_{eee}(\omega_1,\omega_2,\omega_3)\\&=\frac{1}{(2\pi)^{\frac{3}{2}}\sqrt{|b_{1}b_{2}b_{3}|}}\int_{-\infty}^{\infty}\int_{-\infty}^{\infty}\int_{-\infty}^{\infty}f_{eee}^{s_1}(x_1,x_2,x_3)\overline{\Psi_{eee}((x_1+s_1)-\mu_1,x_2-\mu_2,x_3-\mu_3)}\\&\times\cos{\left(\left(s_1\omega_1c_1 - \frac{a_1c_1s_1^2}{2}\right)+\left(\frac{a_1x_1^2}{2b_1}-\frac{x_1(\omega_1-s_1a_1)}{b_1}+\frac{d_1(\omega_1-s_1a_1)^2}{2b_1}-\frac{\pi}{2}\right)\right)}\\&\times\cos{\left(\frac{a_2x_2^2}{2b_2}-\frac{x_2\omega_2}{b_2}+\frac{d_2\omega_2^2}{2b_2}-\frac{\pi}{2}\right)}\cos{\left(\frac{a_3x_3^2}{2b_3}-\frac{x_3\omega_3}{b_3}+\frac{d_3\omega_3^2}{2b_3}-\frac{\pi}{2}\right)}{\rm d}x_1{\rm d}x_2{\rm d}x_3,
	\end{align*}
	setting $\rho_1=\mu_1-s_1, \rho_2=\mu_2$, and $\rho_3=\mu_3$, we have
	\begin{align*}
	&=\frac{1}{(2\pi)^{\frac{3}{2}}\sqrt{|b_{1}b_{2}b_{3}|}}\Bigg[\cos\left(s_1\omega_1c_1 - \frac{a_1c_1s_1^2}{2}\right)\int_{-\infty}^{\infty}\int_{-\infty}^{\infty}\int_{-\infty}^{\infty}f_{eee}(x_1,x_2,x_3)\\&\times\overline{\Psi_{eee}(x_1-\rho_1,x_2-\rho_2,x_3-\rho_3)}\cos\left(\frac{a_1x_1^2}{2b_1}-\frac{x_1(\omega_1-s_1a_1)}{b_1}+\frac{d_1(\omega_1-s_1a_1)^2}{2b_1}-\frac{\pi}{2}\right)\\&\times\cos{\left(\frac{a_2x_2^2}{2b_2}-\frac{x_2\omega_2}{b_2}+\frac{d_2\omega_2^2}{2b_2}-\frac{\pi}{2}\right)}\cos{\left(\frac{a_3x_3^2}{2b_3}-\frac{x_3\omega_3}{b_3}+\frac{d_3\omega_3^2}{2b_3}-\frac{\pi}{2}\right)}{\rm d}x_1{\rm d}x_2{\rm d}x_3\\&-\sin\left(s_1\omega_1c_1-\frac{a_1c_1s_1^2}{2}\right)\int_{-\infty}^{\infty}\int_{-\infty}^{\infty}\int_{-\infty}^{\infty}f_{eee}(x_1,x_2,x_3)\\&\times\overline{\Psi_{eee}(x_1-\rho_1,x_2-\rho_2,x_3-\rho_3)}\sin\left(\frac{a_1x_1^2}{2b_1}-\frac{x_1(\omega_1-s_1a_1)}{b_1}+\frac{d_1(\omega_1-s_1a_1)^2}{2b_1}-\frac{\pi}{2}\right)\\&\times\cos{\left(\frac{a_2x_2^2}{2b_2}-\frac{x_2\omega_2}{b_2}+\frac{d_2\omega_2^2}{2b_2}-\frac{\pi}{2}\right)}\cos{\left(\frac{a_3x_3^2}{2b_3}-\frac{x_3\omega_3}{b_3}+\frac{d_3\omega_3^2}{2b_3}-\frac{\pi}{2}\right)}{\rm d}x_1{\rm d}x_2{\rm d}x_3\Bigg].
	\end{align*}
	Now, we assume that
	\begin{align*}
	&\{\mathcal{G}_{\mathcal{N}_1,\mathcal{N}_2,\mathcal{N}_3}^{\mathbb{O}}(f,\Psi)\}_{see}(\rho_1,\omega_2,\omega_3)=\frac{1}{(2\pi)^{\frac{3}{2}}\sqrt{|b_{1}b_{2}b_{3}|}}\int_{-\infty}^{\infty}\int_{-\infty}^{\infty}\int_{-\infty}^{\infty}f_{eee}(x_1,x_2,x_3)\\&\times\overline{\Psi_{eee}(x_1-\rho_1,x_2-\rho_2,x_3-\rho_3)}\sin\left(\frac{a_1x_1^2}{2b_1}-\frac{x_1(\omega_1-s_1a_1)}{b_1}+\frac{d_1(\omega_1-s_1a_1)^2}{2b_1}-\frac{\pi}{2}\right)\\&\times\cos{\left(\frac{a_2x_2^2}{2b_2}-\frac{x_2\omega_2}{b_2}+\frac{d_2\omega_2^2}{2b_2}-\frac{\pi}{2}\right)}\cos{\left(\frac{a_3x_3^2}{2b_3}-\frac{x_3\omega_3}{b_3}+\frac{d_3\omega_3^2}{2b_3}-\frac{\pi}{2}\right)}{\rm d}x_1{\rm d}x_2{\rm d}x_3,
	\end{align*}
	\begin{align*}
	&\{\mathcal{G}_{\mathcal{N}_1,\mathcal{N}_2,\mathcal{N}_3}^{\mathbb{O}}(f,\Psi)\}_{cee}(\rho_1,\omega_2,\omega_3)=\frac{1}{(2\pi)^{\frac{3}{2}}\sqrt{|b_{1}b_{2}b_{3}|}}\int_{-\infty}^{\infty}\int_{-\infty}^{\infty}\int_{-\infty}^{\infty}f_{oee}(x_1,x_2,x_3)\\&\times\overline{\Psi_{oee}(x_!-\rho_1,x_2-\rho_2,x_3-\rho_3)}\cos\left(\frac{a_1x_1^2}{2b_1}-\frac{x_1(\omega_1-s_1a_1)}{b_1}+\frac{d_1(\omega_1-s_1a_1)^2}{2b_1}-\frac{\pi}{2}\right)\\&\times\cos{\left(\frac{a_2x_2^2}{2b_2}-\frac{x_2\omega_2}{b_2}+\frac{d_2\omega_2^2}{2b_2}-\frac{\pi}{2}\right)}\cos{\left(\frac{a_3x_3^2}{2b_3}-\frac{x_3\omega_3}{b_3}+\frac{d_3\omega_3^2}{2b_3}-\frac{\pi}{2}\right)}{\rm d}x_1{\rm d}x_2{\rm d}x_3,
	\end{align*}
	\begin{align*}
	&\{\mathcal{G}_{\mathcal{N}_1,\mathcal{N}_2,\mathcal{N}_3}^{\mathbb{O}}(f,\Psi)\}_{soe}(\rho_1,\omega_2,\omega_3)=\frac{1}{(2\pi)^{\frac{3}{2}}\sqrt{|b_{1}b_{2}b_{3}|}}\int_{-\infty}^{\infty}\int_{-\infty}^{\infty}\int_{-\infty}^{\infty}f_{eoe}(x_1,x_2,x_3)\\&\times\overline{\Psi_{eoe}(x_1-\rho_1,x_2-\rho_2,x_3-\rho_3)}\sin\left(\frac{a_1x_1^2}{2b_1}-\frac{x_1(\omega_1-s_1a_1)}{b_1}+\frac{d_1(\omega_1-s_1a_1)^2}{2b_1}-\frac{\pi}{2}\right)\\&\times\sin{\left(\frac{a_2x_2^2}{2b_2}-\frac{x_2\omega_2}{b_2}+\frac{d_2\omega_2^2}{2b_2}-\frac{\pi}{2}\right)}\cos{\left(\frac{a_3x_3^2}{2b_3}-\frac{x_3\omega_3}{b_3}+\frac{d_3\omega_3^2}{2b_3}-\frac{\pi}{2}\right)}{\rm d}x_1{\rm d}x_2{\rm d}x_3,
	\end{align*}
	\begin{align*}
	&\{\mathcal{G}_{\mathcal{N}_1,\mathcal{N}_2,\mathcal{N}_3}^{\mathbb{O}}(f,\Psi)\}_{coe}(\rho_1,\omega_2,\omega_3)=\frac{1}{(2\pi)^{\frac{3}{2}}\sqrt{|b_{1}b_{2}b_{3}|}}\int_{-\infty}^{\infty}\int_{-\infty}^{\infty}\int_{-\infty}^{\infty}f_{ooe}(x_1,x_2,x_3)\\&\times\overline{\Psi_{ooe}(x_1-\rho_1,x_2-\rho_2,x_3-\rho_3)}\cos\left(\frac{a_1x_1^2}{2b_1}-\frac{x_1(\omega_1-s_1a_1)}{b_1}+\frac{d_1(\omega_1-s_1a_1)^2}{2b_1}-\frac{\pi}{2}\right)\\&\times\sin{\left(\frac{a_2x_2^2}{2b_2}-\frac{x_2\omega_2}{b_2}+\frac{d_2\omega_2^2}{2b_2}-\frac{\pi}{2}\right)}\cos{\left(\frac{a_3x_3^2}{2b_3}-\frac{x_3\omega_3}{b_3}+\frac{d_3\omega_3^2}{2b_3}-\frac{\pi}{2}\right)}{\rm d}x_1{\rm d}x_2{\rm d}x_3,
	\end{align*}
	\begin{align*}
	&\{\mathcal{G}_{\mathcal{N}_1,\mathcal{N}_2,\mathcal{N}_3}^{\mathbb{O}}(f,\Psi)\}_{seo}(\rho_1,\omega_2,\omega_3)=\frac{1}{(2\pi)^{\frac{3}{2}}\sqrt{|b_{1}b_{2}b_{3}|}}
	\int_{-\infty}^{\infty}\int_{-\infty}^{\infty}\int_{-\infty}^{\infty}f_{oeo}(x_1,x_2,x_3)\\&\times\overline{\Psi_{oeo}(x_1-\rho_1,x_2-\rho_2,x_3-\rho_3)}\sin\left(\frac{a_1x_1^2}{2b_1}-\frac{x_1(\omega_1-s_1a_1)}{b_1}+\frac{d_1(\omega_1-s_1a_1)^2}{2b_1}-\frac{\pi}{2}\right)\\&\times\cos{\left(\frac{a_2x_2^2}{2b_2}-\frac{x_2\omega_2}{b_2}+\frac{d_2\omega_2^2}{2b_2}-\frac{\pi}{2}\right)}\sin{\left(\frac{a_3x_3^2}{2b_3}-\frac{x_3\omega_3}{b_3}+\frac{d_3\omega_3^2}{2b_3}-\frac{\pi}{2}\right)}{\rm d}x_1{\rm d}x_2{\rm d}x_3,
	\end{align*}
	\begin{align*}
	&\{\mathcal{G}_{\mathcal{N}_1,\mathcal{N}_2,\mathcal{N}_3}^{\mathbb{O}}(f,\Psi)\}_{ceo}(\rho_1,\omega_2,\omega_3)=\frac{1}{(2\pi)^{\frac{3}{2}}\sqrt{|b_{1}b_{2}b_{3}|}}\int_{-\infty}^{\infty}\int_{-\infty}^{\infty}\int_{-\infty}^{\infty}f_{eeo}(x_1,x_2,x_3)\\&\times\overline{\Psi_{eeo}(x_1-\rho_1,x_2-\rho_2,x_3-\rho_3)}\cos\left(\frac{a_1x_1^2}{2b_1}-\frac{x_1(\omega_1-s_1a_1)}{b_1}+\frac{d_1(\omega_1-s_1a_1)^2}{2b_1}-\frac{\pi}{2}\right)\\&\times\cos{\left(\frac{a_2x_2^2}{2b_2}-\frac{x_2\omega_2}{b_2}+\frac{d_2\omega_2^2}{2b_2}-\frac{\pi}{2}\right)}\sin{\left(\frac{a_3x_3^2}{2b_3}-\frac{x_3\omega_3}{b_3}+\frac{d_3\omega_3^2}{2b_3}-\frac{\pi}{2}\right)}{\rm d}x_1{\rm d}x_2{\rm d}x_3,
	\end{align*}
	\begin{align*}
	&\{\mathcal{G}_{\mathcal{N}_1,\mathcal{N}_2,\mathcal{N}_3}^{\mathbb{O}}(f,\Psi)\}_{soo}(\rho_1,\omega_2,\omega_3)=\frac{1}{(2\pi)^{\frac{3}{2}}\sqrt{|b_{1}b_{2}b_{3}|}}\int_{-\infty}^{\infty}\int_{-\infty}^{\infty}\int_{-\infty}^{\infty}f_{eoo}(x_1,x_2,x_3)\\&\times\overline{\Psi_{eoo}(x_1-\rho_1,x_2-\rho_2,x_3-\rho_3)}\sin\left(\frac{a_1x_1^2}{2b_1}-\frac{x_1(\omega_1-s_1a_1)}{b_1}+\frac{d_1(\omega_1-s_1a_1)^2}{2b_1}-\frac{\pi}{2}\right)\\&\times\sin{\left(\frac{a_2x_2^2}{2b_2}-\frac{x_2\omega_2}{b_2}+\frac{d_2\omega_2^2}{2b_2}-\frac{\pi}{2}\right)}\sin{\left(\frac{a_3x_3^2}{2b_3}-\frac{x_3\omega_3}{b_3}+\frac{d_3\omega_3^2}{2b_3}-\frac{\pi}{2}\right)}{\rm d}x_1{\rm d}x_2{\rm d}x_3,
	\end{align*}
	and
	\begin{align*}
	&\{\mathcal{G}_{\mathcal{N}_1,\mathcal{N}_2,\mathcal{N}_3}^{\mathbb{O}}(f,\Psi)\}_{coo}(\rho_1,\omega_2,\omega_3)=\frac{1}{(2\pi)^{\frac{3}{2}}\sqrt{|b_{1}b_{2}b_{3}|}}\int_{-\infty}^{\infty}\int_{-\infty}^{\infty}\int_{-\infty}^{\infty}f_{ooo}(x_1,x_2,x_3)\\&\times\overline{\Psi_{ooo}(x_1-\rho_1,x_2-\rho_2,x_3-\rho_3)}\cos\left(\frac{a_1x_1^2}{2b_1}-\frac{x_1(\omega_1-s_1a_1)}{b_1}+\frac{d_1(\omega_1-s_1a_1)^2}{2b_1}-\frac{\pi}{2}\right)\\&\times\sin{\left(\frac{a_2x_2^2}{2b_2}-\frac{x_2\omega_2}{b_2}+\frac{d_2\omega_2^2}{2b_2}-\frac{\pi}{2}\right)}\sin{\left(\frac{a_3x_3^2}{2b_3}-\frac{x_3\omega_3}{b_3}+\frac{d_3\omega_3^2}{2b_3}-\frac{\pi}{2}\right)}{\rm d}x_1{\rm d}x_2{\rm d}x_3,
	\end{align*}
	then, we have
	\begin{align*}
	&\{\mathcal{G}_{\mathcal{N}_1,\mathcal{N}_2,\mathcal{N}_3}^{\mathbb{O},s_1}(f,\Psi)\}_{eee}(\omega_1,\omega_2,\omega_3)=\cos\left(s_1\omega_1c_1 - \frac{a_1c_1s_1^2}{2}\right)\{\mathcal{G}_{\mathcal{N}_1,\mathcal{N}_2,\mathcal{N}_3}^{\mathbb{O}}(f,\Psi)\}_{eee}\\&\times(\rho_1,\omega_2,\omega_3)-\sin\left(s_1\omega_1c_1 - \frac{a_1c_1s_1^2}{2}\right)\{\mathcal{G}_{\mathcal{N}_1,\mathcal{N}_2,\mathcal{N}_3}^{\mathbb{O}}(f,\Psi)\}_{see}(\rho_1,\omega_2,\omega_3),
	\end{align*}
	where $\rho_1=\mu_1-s_1$. Similarly, we can prove that
	\begin{align*}
	&\{\mathcal{G}_{\mathcal{N}_1,\mathcal{N}_2,\mathcal{N}_3}^{\mathbb{O},s_1}(f,\Psi)\}_{oee}(\omega_1,\omega_2,\omega_3)\\&=\frac{1}{(2\pi)^{\frac{3}{2}}\sqrt{|b_{1}b_{2}b_{3}|}}\int_{-\infty}^{\infty}\int_{-\infty}^{\infty}\int_{-\infty}^{\infty}f_{oee}(x_1,x_2,x_3)\overline{\Psi_{oee}(x_1-\rho_1,x_2-\rho_2,x_3-\rho_3)}\\&\times\sin{\left(\left(s_1\omega_1c_1 - \frac{a_1c_1s_1^2}{2}\right)+\left(\frac{a_1x_1^2}{2b_1}-\frac{x_1(\omega_1-s_1a_1)}{b_1}+\frac{d_1(\omega_1-s_1a_1)^2}{2b_1}-\frac{\pi}{2}\right)\right)}\\&\times\cos{\left(\frac{a_2x_2^2}{2b_2}-\frac{x_2\omega_2}{b_2}+\frac{d_2\omega_2^2}{2b_2}-\frac{\pi}{2}\right)}\cos{\left(\frac{a_3x_3^2}{2b_3}-\frac{x_3\omega_3}{b_3}+\frac{d_3\omega_3^2}{2b_3}-\frac{\pi}{2}\right)}{\rm d}x_1{\rm d}x_2{\rm d}x_3\\&=\frac{1}{(2\pi)^{\frac{3}{2}}\sqrt{|b_{1}b_{2}b_{3}|}}\Bigg[\sin\left(s_1\omega_1c_1 - \frac{a_1c_1s_1^2}{2}\right)\int_{-\infty}^{\infty}\int_{-\infty}^{\infty}\int_{-\infty}^{\infty}f_{oee}(x_1,x_2,x_3)\overline{\Psi_{oee}(x_1}\\&\times\overline{-\rho_1,x_2-\rho_2,x_3-\rho_3)}\cos\left(\frac{a_1x_1^2}{2b_1}-\frac{x_1(\omega_1-s_1a_1)}{b_1}+\frac{d_1(\omega_1-s_1a_1)^2}{2b_1}-\frac{\pi}{2}\right)\\&\times\cos{\left(\frac{a_2x_2^2}{2b_2}-\frac{x_2\omega_2}{b_2}+\frac{d_2\omega_2^2}{2b_2}-\frac{\pi}{2}\right)}\cos{\left(\frac{a_3x_3^2}{2b_3}-\frac{x_3\omega_3}{b_3}+\frac{d_3\omega_3^2}{2b_3}-\frac{\pi}{2}\right)}{\rm d}x_1{\rm d}x_2{\rm d}x_3\\&+\cos\left(s_1\omega_1c_1-\frac{a_1c_1s_1^2}{2}\right)\int_{-\infty}^{\infty}\int_{-\infty}^{\infty}\int_{-\infty}^{\infty}f_{oee}(x_1,x_2,x_3)\overline{\Psi_{oee}(x_1}\\&\times\overline{-\rho_1,x_2-\rho_2,x_3-\rho_3)}\sin\left(\frac{a_1x_1^2}{2b_1}-\frac{x_1(\omega_1-s_1a_1)}{b_1}+\frac{d_1(\omega_1-s_1a_1)^2}{2b_1}-\frac{\pi}{2}\right)\\&\times\cos{\left(\frac{a_2x_2^2}{2b_2}-\frac{x_2\omega_2}{b_2}+\frac{d_2\omega_2^2}{2b_2}-\frac{\pi}{2}\right)}\cos{\left(\frac{a_3x_3^2}{2b_3}-\frac{x_3\omega_3}{b_3}+\frac{d_3\omega_3^2}{2b_3}-\frac{\pi}{2}\right)}{\rm d}x_1{\rm d}x_2{\rm d}x_3\Bigg].
	\end{align*}
	\begin{align*}
	=&\cos\left(s_1\omega_1c_1 - \frac{a_1c_1s_1^2}{2}\right)\{\mathcal{G}_{\mathcal{N}_1,\mathcal{N}_2,\mathcal{N}_3}^{\mathbb{O}}(f,\Psi)\}_{oee}(\rho_1,\omega_2,\omega_3)+\sin\left(s_1\omega_1c_1 - \frac{a_1c_1s_1^2}{2}\right)\\&\times\{\mathcal{G}_{\mathcal{N}_1,\mathcal{N}_2,\mathcal{N}_3}^{\mathbb{O}}(f,\Psi)\}_{cee}(\rho_1,\omega_2,\omega_3).
	\end{align*}
	Summarizing, we have
	\begin{align*}
	\{\mathcal{G}_{\mathcal{N}_1,\mathcal{N}_2,\mathcal{N}_3}^{\mathbb{O},s_1}(f,\Psi)\}_{eee}(\omega_1,\omega_2,\omega_3)&=\cos\left(s_1\omega_1c_1 - \frac{a_1c_1s_1^2}{2}\right)\{\mathcal{G}_{\mathcal{N}_1,\mathcal{N}_2,\mathcal{N}_3}^{\mathbb{O}}(f,\Psi)\}_{eee}\\&\times(\rho_1,\omega_2,\omega_3)-\sin\left(s_1\omega_1c_1 - \frac{a_1c_1s_1^2}{2}\right)\\&\times\{\mathcal{G}_{\mathcal{N}_1,\mathcal{N}_2,\mathcal{N}_3}^{\mathbb{O}}(f,\Psi)\}_{see}(\rho_1,\omega_2,\omega_3),
	\end{align*}
	\begin{align*}
	\{\mathcal{G}_{\mathcal{N}_1,\mathcal{N}_2,\mathcal{N}_3}^{\mathbb{O},s_1}(f,\Psi)\}_{oee}(\omega_1,\omega_2,\omega_3)&=\cos\left(s_1\omega_1c_1 - \frac{a_1c_1s_1^2}{2}\right)\{\mathcal{G}_{\mathcal{N}_1,\mathcal{N}_2,\mathcal{N}_3}^{\mathbb{O}}(f,\Psi)\}_{oee}\\&\times(\rho_1,\omega_2,\omega_3)+\sin\left(s_1\omega_1c_1 - \frac{a_1c_1s_1^2}{2}\right)\\&\times\{\mathcal{G}_{\mathcal{N}_1,\mathcal{N}_2,\mathcal{N}_3}^{\mathbb{O}}(f,\Psi)\}_{cee}(\rho_1,\omega_2,\omega_3),
	\end{align*}
	\begin{align*}
	\{\mathcal{G}_{\mathcal{N}_1,\mathcal{N}_2,\mathcal{N}_3}^{\mathbb{O},s_1}(f,\Psi)\}_{eoe}(\omega_1,\omega_2,\omega_3)&=\cos\left(s_1\omega_1c_1 - \frac{a_1c_1s_1^2}{2}\right)\{\mathcal{G}_{\mathcal{N}_1,\mathcal{N}_2,\mathcal{N}_3}^{\mathbb{O}}(f,\Psi)\}_{eoe}\\&\times(\rho_1,\omega_2,\omega_3)-\sin\left(s_1\omega_1c_1 - \frac{a_1c_1s_1^2}{2}\right)\\&\times\{\mathcal{G}_{\mathcal{N}_1,\mathcal{N}_2,\mathcal{N}_3}^{\mathbb{O}}(f,\Psi)\}_{soe}(\rho_1,\omega_2,\omega_3),
	\end{align*}
	\begin{align*}
	\{\mathcal{G}_{\mathcal{N}_1,\mathcal{N}_2,\mathcal{N}_3}^{\mathbb{O},s_1}(f,\Psi)\}_{ooe}(\omega_1,\omega_2,\omega_3)&=\cos\left(s_1\omega_1c_1 - \frac{a_1c_1s_1^2}{2}\right)\{\mathcal{G}_{\mathcal{N}_1,\mathcal{N}_2,\mathcal{N}_3}^{\mathbb{O}}(f,\Psi)\}_{ooe}\\&\times(\rho_1,\omega_2,\omega_3)+\sin\left(s_1\omega_1c_1 - \frac{a_1c_1s_1^2}{2}\right)\\&\times\{\mathcal{G}_{\mathcal{N}_1,\mathcal{N}_2,\mathcal{N}_3}^{\mathbb{O}}(f,\Psi)\}_{coe}(\rho_1,\omega_2,\omega_3),
	\end{align*}
	\begin{align*}
	\{\mathcal{G}_{\mathcal{N}_1,\mathcal{N}_2,\mathcal{N}_3}^{\mathbb{O},s_1}(f,\Psi)\}_{eeo}(\omega_1,\omega_2,\omega_3)&=\cos\left(s_1\omega_1c_1 - \frac{a_1c_1s_1^2}{2}\right)\{\mathcal{G}_{\mathcal{N}_1,\mathcal{N}_2,\mathcal{N}_3}^{\mathbb{O}}(f,\Psi)\}_{eeo}\\&\times(\rho_1,\omega_2,\omega_3)-\sin\left(s_1\omega_1c_1 - \frac{a_1c_1s_1^2}{2}\right)\\&\times\{\mathcal{G}_{\mathcal{N}_1,\mathcal{N}_2,\mathcal{N}_3}^{\mathbb{O}}(f,\Psi)\}_{seo}(\rho_1,\omega_2,\omega_3),
	\end{align*}
	\begin{align*}
	\{\mathcal{G}_{\mathcal{N}_1,\mathcal{N}_2,\mathcal{N}_3}^{\mathbb{O},s_1}(f,\Psi)\}_{oeo}(\omega_1,\omega_2,\omega_3)&=\cos\left(s_1\omega_1c_1 - \frac{a_1c_1s_1^2}{2}\right)\{\mathcal{G}_{\mathcal{N}_1,\mathcal{N}_2,\mathcal{N}_3}^{\mathbb{O}}(f,\Psi)\}_{oeo}\\&\times(\rho_1,\omega_2,\omega_3)+\sin\left(s_1\omega_1c_1 - \frac{a_1c_1s_1^2}{2}\right)\\&\times\{\mathcal{G}_{\mathcal{N}_1,\mathcal{N}_2,\mathcal{N}_3}^{\mathbb{O}}(f,\Psi)\}_{ceo}(\rho_1,\omega_2,\omega_3),
	\end{align*}
	\begin{align*}
	\{\mathcal{G}_{\mathcal{N}_1,\mathcal{N}_2,\mathcal{N}_3}^{\mathbb{O},s_1}(f,\Psi)\}_{eoo}(\omega_1,\omega_2,\omega_3)&=\cos\left(s_1\omega_1c_1 - \frac{a_1c_1s_1^2}{2}\right)\{\mathcal{G}_{\mathcal{N}_1,\mathcal{N}_2,\mathcal{N}_3}^{\mathbb{O}}(f,\Psi)\}_{eoo}\\&\times(\rho_1,\omega_2,\omega_3)-\sin\left(s_1\omega_1c_1 - \frac{a_1c_1s_1^2}{2}\right)\\&\times\{\mathcal{G}_{\mathcal{N}_1,\mathcal{N}_2,\mathcal{N}_3}^{\mathbb{O}}(f,\Psi)\}_{soo}(\rho_1,\omega_2,\omega_3),
	\end{align*}
	and
	\begin{align*}
	\{\mathcal{G}_{\mathcal{N}_1,\mathcal{N}_2,\mathcal{N}_3}^{\mathbb{O},s_1}(f,\Psi)\}_{ooo}(\omega_1,\omega_2,\omega_3)&=\cos\left(s_1\omega_1c_1 - \frac{a_1c_1s_1^2}{2}\right)\{\mathcal{G}_{\mathcal{N}_1,\mathcal{N}_2,\mathcal{N}_3}^{\mathbb{O}}(f,\Psi)\}_{ooo}\\&\times(\rho_1,\omega_2,\omega_3)+\sin\left(s_1\omega_1c_1 - \frac{a_1c_1s_1^2}{2}\right)\\&\times\{\mathcal{G}_{\mathcal{N}_1,\mathcal{N}_2,\mathcal{N}_3}^{\mathbb{O}}(f,\Psi)\}_{coo}(\rho_1,\omega_2,\omega_3).
	\end{align*}
	Using from \eqref{eq:3.66}, we obtain
	\begin{align*}
	\mathcal{G}_{\mathcal{N}_1,\mathcal{N}_2,\mathcal{N}_3}^{\mathbb{O},s_1}(f,\Psi)(\omega_{1},\omega_{2},\omega_{3})&=\cos\left(s_1\omega_1c_1 - \frac{a_1c_1s_1^2}{2}\right)	\{\mathcal{G}_{\mathcal{N}_1,\mathcal{N}_2,\mathcal{N}_3}^{\mathbb{O}}(f,\Psi)\}(\rho_1,\omega_2,\omega_3)\\&-\sin\left(s_1\omega_1c_1 - \frac{a_1c_1s_1^2}{2}\right)\Delta_1 f(\rho_1,\omega_2,\omega_3),
	\end{align*}
	where
	\begin{align*}
	\Delta_1 f&=	\{\mathcal{G}_{\mathcal{N}_1,\mathcal{N}_2,\mathcal{N}_3}^{\mathbb{O}}(f,\Psi)\}_{see}-\{\mathcal{G}_{\mathcal{N}_1,\mathcal{N}_2,\mathcal{N}_3}^{\mathbb{O}}(f,\Psi)\}_{cee}e_1+\{\mathcal{G}_{\mathcal{N}_1,\mathcal{N}_2,\mathcal{N}_3}^{\mathbb{O}}(f,\Psi)\}_{soe}e_2\nonumber\\&-\{\mathcal{G}_{\mathcal{N}_1,\mathcal{N}_2,\mathcal{N}_3}^{\mathbb{O}}(f,\Psi)\}_{coe}e_3+\{\mathcal{G}_{\mathcal{N}_1,\mathcal{N}_2,\mathcal{N}_3}^{\mathbb{O}}(f,\Psi)\}_{seo}e_4-\{\mathcal{G}_{\mathcal{N}_1,\mathcal{N}_2,\mathcal{N}_3}^{\mathbb{O}}(f,\Psi)\}_{ceo}e_5\nonumber \\&+\{\mathcal{G}_{\mathcal{N}_1,\mathcal{N}_2,\mathcal{N}_3}^{\mathbb{O}}(f,\Psi)\}_{soo}e_6-\{\mathcal{G}_{\mathcal{N}_1,\mathcal{N}_2,\mathcal{N}_3}^{\mathbb{O}}(f,\Psi)\}_{coo}e_7.
	\end{align*}
	Similarly, one can prove the shifting property with respect to the variables $t_2$ and $t_3$. Hence, we get the desired result. 
\end{proof}
%\begin{definition}
%		(The Riemann-Lebegue lemma of the 3D WOCLCT) 	Let the WOCLCT of an octonion-valued signal (or function) $f(t_1,t_2,t_3)$ is a map from $\mathbb{R}^3$ to $\mathbb{O}$ with respect to OW signal (or function) $\Psi \in L^2(\mathbb{R}^3,\mathbb{O})\setminus\{0\}$, then $	\{\mathcal{G}_{\mathcal{N}_1,\mathcal{N}_2,\mathcal{N}_3}^{\mathbb{O}}(f,\Psi)\}((\omega_1,\omega_2,\omega_3),(\mu_1,\mu_2,\mu_3))$ satisfies
%		\begin{eqnarray*}
%			\lim_{|\omega_{1}|\rightarrow\infty}	\{\mathcal{G}_{\mathcal{N}_1,\mathcal{N}_2,\mathcal{N}_3}^{\mathbb{O}}(f,\Psi)\}((\omega_1,\omega_2,\omega_3),(\mu_1,\mu_2,\mu_3))(\omega_1,\omega_2,\omega_3)=0\\
%			\lim_{|\omega_{2}|\rightarrow\infty}	\{\mathcal{G}_{\mathcal{N}_1,\mathcal{N}_2,\mathcal{N}_3}^{\mathbb{O}}(f,\Psi)\}((\omega_1,\omega_2,\omega_3),(\mu_1,\mu_2,\mu_3))(\omega_1,\omega_2,\omega_3)=0\\
%			\lim_{|\omega_{3}|\rightarrow\infty}	\{\mathcal{G}_{\mathcal{N}_1,\mathcal{N}_2,\mathcal{N}_3}^{\mathbb{O}}(f,\Psi)\}((\omega_1,\omega_2,\omega_3),(\mu_1,\mu_2,\mu_3))(\omega_1,\omega_2,\omega_3)=0
%		\end{eqnarray*}
%\end{definition}
\section{Sharp inequalities and the associated Uncertainty principles for the 3D WOCLCT}\label{sec:4}
In this section, we focus proving the main results, such as sharp Pitt's inequality, sharp Young-Hausdorff inequality, logarithmic uncertainty principle, Heisenberg's uncertainty principle, and Donoho-Stark's uncertainty principle. 
\begin{proposition}\label{pro:4.1}[Sharp Pitt's inequality for the 3D OCLCT]\label{thm:3.6} For $f(t_1,t_2,t_3) \in \mathbb{S}(\mathbb{R}^3,\mathbb{O}),$ $0\leq \beta < 3,$
	\begin{align}\label{eq:3.8}
	&\int_{-\infty}^{\infty}\int_{-\infty}^{\infty}\int_{-\infty}^{\infty}\Big|(\omega_1,\omega_2,\omega_3)\Big|^{-\beta}\left|(\mathcal{L}^{\mathbb{O}}_{\mathcal{N}_{1},\mathcal{N}_{2},\mathcal{N}_{3}}f)(\omega_1,\omega_2,\omega_3)\right|^2 {\rm d}\omega_1{\rm d}\omega_2{\rm d}\omega_3 \nonumber \\& \leq\frac{M_{\beta}}{2 \pi |b_3||b_1b_2|^{\beta} }\int_{-\infty}^{\infty}\int_{-\infty}^{\infty}\int_{-\infty}^{\infty}\Big|(t_1,t_2,t_3)\Big|^{\beta}\Big|f(t_1,t_2,t_3)\Big|^2{\rm d}t_1{\rm d}t_2{\rm d}t_3,
	\end{align}
	where $M_{\beta}= \left(\frac{\Gamma(\frac{3-\beta}{4})}{\Gamma(\frac{3+\beta}{4})}\right)^2$.
\end{proposition}
\begin{theorem}[Sharp Pitt's inequality for the 3D WOCLCT]
	For $f(t_1,t_2,t_3) \in \mathbb{S}(\mathbb{R}^3,\mathbb{O})$ with respect to OW signal (or function)  $\Psi \in L^2(\mathbb{R}^3,\mathbb{O})\setminus\{0\}$, $0\leq \beta < 3,$ then we have 
	\begin{align}\label{eq:4.2}
	&\int_{-\infty}^{\infty}\int_{-\infty}^{\infty}\int_{-\infty}^{\infty}\int_{-\infty}^{\infty}\int_{-\infty}^{\infty}\int_{-\infty}^{\infty}\Big|(\omega_1,\omega_2,\omega_3)\Big|^{-\beta}\left|\{\mathcal{G}_{\mathcal{N}_1,\mathcal{N}_2,\mathcal{N}_3}^{\mathbb{O}}(f,\Psi)\}(\omega_1,\omega_2,\omega_3)\right|^2 {\rm d}\omega_1{\rm d}\omega_2 \nonumber\\&\times{\rm d}\omega_3{\rm d}\mu_1{\rm d}\mu_2{\rm d}\mu_3 \leq  \frac{M_{\beta}||\Psi||^2_2}{2 \pi |b_3||b_1b_2|^{\beta} }\int_{-\infty}^{\infty}\int_{-\infty}^{\infty}\int_{-\infty}^{\infty}\Big|(t_1,t_2,t_3)\Big|^{\beta}\Big|f(t_1,t_2,t_3)\Big|^2{\rm d}t_1{\rm d}t_2{\rm d}t_3.
	\end{align}
\end{theorem}
\begin{proof}
	Replace $f(t_1,t_2,t_3)$ by $f(t_1,t_2,t_3)\overline{\Psi(t_1-\mu_1,t_2-\mu_2,t_3-\mu_3)}$ in Proposition \ref{pro:4.1} and integrate with respect to $\mu_1,\mu_2$, and $\mu_3$, we get
	\begin{align*}
	&\int_{-\infty}^{\infty}\int_{-\infty}^{\infty}\int_{-\infty}^{\infty}\int_{-\infty}^{\infty}\int_{-\infty}^{\infty}\int_{-\infty}^{\infty}\Big|(\omega_1,\omega_2,\omega_3)\Big|^{-\beta}\left|\{\mathcal{G}_{\mathcal{N}_1,\mathcal{N}_2,\mathcal{N}_3}^{\mathbb{O}}(f,\Psi)\}(\omega_1,\omega_2,\omega_3)\right|^2{\rm d}\omega_1{\rm d}\omega_2\\&\times{\rm d}\omega_3{\rm d}\mu_1{\rm d}\mu_2{\rm d}\mu_3\leq \frac{M_{\beta}}{2 \pi |b_3||b_1b_2|^{\beta} }\int_{-\infty}^{\infty}\int_{-\infty}^{\infty}\int_{-\infty}^{\infty}\int_{-\infty}^{\infty}\int_{-\infty}^{\infty}\int_{-\infty}^{\infty}\Big|(t_1,t_2,t_3)\Big|^{\beta}\Big|f(t_1,t_2,t_3)\\&\times\overline{\Psi(t_1-\mu_1,t_2-\mu_2,t_3-\mu_3)}\Big|^2{\rm d}t_1{\rm d}t_2{\rm d}t_3{\rm d}\mu_1{\rm d}\mu_2{\rm d}\mu_3	\\& \leq \frac{M_{\beta}}{2 \pi |b_3||b_1b_2|^{\beta} }\int_{-\infty}^{\infty}\int_{-\infty}^{\infty}\int_{-\infty}^{\infty}\int_{-\infty}^{\infty}\int_{-\infty}^{\infty}\int_{-\infty}^{\infty}\Big|(t_1,t_2,t_3)\Big|^{\beta}\\&\times\Big|f(t_1,t_2,t_3)\Big|^2|{\Psi(t_1-\mu_1,t_2-\mu_2,t_3-\mu_3)}|^2{\rm d}t_1{\rm d}t_2{\rm d}t_3{\rm d}\mu_1{\rm d}\mu_2{\rm d}\mu_3	\\& \leq \frac{M_{\beta}||\Psi||^2_2}{2 \pi |b_3||b_1b_2|^{\beta} }\int_{-\infty}^{\infty}\int_{-\infty}^{\infty}\int_{-\infty}^{\infty}\Big|(t_1,t_2,t_3)\Big|^{\beta}\Big|f(t_1,t_2,t_3)\Big|^2{\rm d}t_1{\rm d}t_2{\rm d}t_3.
	\end{align*}
	Hence, we obtain the desired result.
\end{proof}
\begin{theorem}(Logarithmic uncertainty principle for 3D WOCLCT) For $f(t_1,t_2,t_3) \in \mathbb{S}(\mathbb{R}^3,\mathbb{O})$ with respect to OW signal (or function)  $\Psi \in L^2(\mathbb{R}^3,\mathbb{O})\setminus\{0\}$, then we have
	\begin{align*} 
	&2 \pi |b_3|\int_{-\infty}^{\infty}\int_{-\infty}^{\infty}\int_{-\infty}^{\infty}\int_{-\infty}^{\infty}\int_{-\infty}^{\infty}\int_{-\infty}^{\infty} \ln\Big|(\omega_1,\omega_2,\omega_3)\Big|\Big|\{\mathcal{G}_{\mathcal{N}_1,\mathcal{N}_2,\mathcal{N}_3}^{\mathbb{O}}(f,\Psi)\}(\omega_1,\omega_2,\omega_3)\Big|^2\\&\times{\rm d}\omega_1{\rm d}\omega_2{\rm d}\omega_3{\rm d}\mu_1{\rm d}\mu_2{\rm d}\mu_3+||\Psi||^2_2 \int_{-\infty}^{\infty}\int_{-\infty}^{\infty}\int_{-\infty}^{\infty} \ln\Big|(t_1,t_2,t_3)\Big|\Big|f(t_1,t_2,t_3)\Big|^2{\rm d}t_1{\rm d}t_2{\rm d}t_3 \\&\geq K_{0}^{'}||\Psi||^2_2 \int_{-\infty}^{\infty}\int_{-\infty}^{\infty}\int_{-\infty}^{\infty}\Big|f(t_1,t_2,t_3)\Big|^2{\rm d}t_1{\rm d}t_2{\rm d}t_3,
	\end{align*}
	where $K_{0}^{'} = \frac{\rm d}{{\rm d}\beta}\left(\frac{-M_{\beta}}{|b_1b_2|^{\beta}}\right)$ at ${\beta=0}$.
\end{theorem}
\begin{proof}We can rewrite the equation \eqref{eq:4.2} in the following form
	\begin{align*}
	&2 \pi |b_3|\int_{-\infty}^{\infty}\int_{-\infty}^{\infty}\int_{-\infty}^{\infty}\int_{-\infty}^{\infty}\int_{-\infty}^{\infty}\int_{-\infty}^{\infty}\Big|(\omega_1,\omega_2,\omega_3)\Big|^{-\beta}\left|\{\mathcal{G}_{\mathcal{N}_1,\mathcal{N}_2,\mathcal{N}_3}^{\mathbb{O}}(f,\Psi)\}(\omega_1,\omega_2,\omega_3)\right|^2 \\&\times{\rm d}\omega_1{\rm d}\omega_2{\rm d}\omega_3{\rm d}\mu_1{\rm d}\mu_2{\rm d}\mu_3\leq \frac{M_{\beta}||\Psi||_2^2}{|b_1b_2|^{\beta} }\int_{-\infty}^{\infty}\int_{-\infty}^{\infty}\int_{-\infty}^{\infty}\Big|(t_1,t_2,t_3)\Big|^{\beta}\Big|f(t_1,t_2,t_3)\Big|^2{\rm d}t_1{\rm d}t_2{\rm d}t_3.
	\end{align*}
	Now, for every $0 \leq \beta <3$, define
	\begin{align}\label{eq:3.22}
	&	C(\beta)=\nonumber\\ &2 \pi |b_3|\int_{-\infty}^{\infty}\int_{-\infty}^{\infty}\int_{-\infty}^{\infty}\int_{-\infty}^{\infty}\int_{-\infty}^{\infty}\int_{-\infty}^{\infty}\Big|(\omega_1,\omega_2,\omega_3)\Big|^{-\beta}\left|\{\mathcal{G}_{\mathcal{N}_1,\mathcal{N}_2,\mathcal{N}_3}^{\mathbb{O}}(f,\Psi)\}(\omega_1,\omega_2,\omega_3)\right|^2 \nonumber \\&\times{\rm d}\omega_1{\rm d}\omega_2{\rm d}\omega_3{\rm d}\mu_1{\rm d}\mu_2{\rm d}\mu_3 -\frac{M_{\beta}||\Psi||_2^2}{|b_1b_2|^{\beta} }\int_{-\infty}^{\infty}\int_{-\infty}^{\infty}\int_{-\infty}^{\infty}\Big|(t_1,t_2,t_3)\Big|^{\beta}\Big|f(t_1,t_2,t_3)\Big|^2{\rm d}t_1{\rm d}t_2{\rm d}t_3 \leq 0.
	\end{align}
	Differentiating equation \eqref{eq:3.22} with respect to $\beta$, we have
	\begin{align}\label{eq:3.23}
	&C^{'}(\beta)  = - 2 \pi |b_3|\int_{-\infty}^{\infty}\int_{-\infty}^{\infty}\int_{-\infty}^{\infty}\int_{-\infty}^{\infty}\int_{-\infty}^{\infty}\int_{-\infty}^{\infty}\Big|(\omega_1,\omega_2,\omega_3)\Big|^{-\beta} \ln\Big|(\omega_1,\omega_2,\omega_3)\Big|\nonumber\\&\times\left|\{\mathcal{G}_{\mathcal{N}_1,\mathcal{N}_2,\mathcal{N}_3}^{\mathbb{O}}(f,\Psi)\}(\omega_1,\omega_2,\omega_3)\right|^2 {\rm d}\omega_1{\rm d}\omega_2{\rm d}\omega_3{\rm d}\mu_1{\rm d}\mu_2{\rm d}\mu_3 -E_{\beta}||\Psi||^2_2\int_{-\infty}^{\infty}\int_{-\infty}^{\infty}\int_{-\infty}^{\infty}\nonumber\\&\times\Big|(t_1,t_2,t_3)\Big|^{\beta} \ln\Big|(t_1,t_2,t_3)\Big|\Big|f(t_1,t_2,t_3)\Big|^2{\rm d}t_1{\rm d}t_2{\rm d}t_3-E_{\beta}^{'}||\Psi||^2_2\int_{-\infty}^{\infty}\int_{-\infty}^{\infty}\int_{-\infty}^{\infty}\nonumber\\&\times\Big|(t_1,t_2,t_3)\Big|^{\beta}\Big|f(t_1,t_2,t_3)\Big|^2{\rm d}t_1{\rm d}t_2{\rm d}t_3\leq 0,
	\end{align}
	where
	\begin{align}\label{eq:3.24}
	&E_{\beta} = \frac{M_{\beta}}{|b_1b_2|^{\beta} } \ \ \text {and} \  \ E_{\beta}^{'} = \frac{{M}_{\beta}^{'}-\ln |b_1b_2|}{|b_1b_2|^{\beta}}  \nonumber\\ & {M}_{\beta}^{'} = \frac{\Gamma(\frac{3+\beta}{4})\Gamma(\frac{3-\beta}{4})\Gamma^{'}(\frac{3-\beta}{4})+ \left(\Gamma(\frac{3-\beta}{4})\right)^2\Gamma^{'}(\frac{3+\beta}{4})}{\left(\Gamma(\frac{3+\beta}{4})\right)^3}.
	\end{align}
	Setting $ \beta = 0$ in equation \eqref{eq:3.23} and equation \eqref{eq:3.24}, we get
	\begin{align*}
	&2 \pi |b_3|\int_{-\infty}^{\infty}\int_{-\infty}^{\infty}\int_{-\infty}^{\infty}\int_{-\infty}^{\infty}\int_{-\infty}^{\infty}\int_{-\infty}^{\infty} \ln\Big|(\omega_1,\omega_2,\omega_3)\Big|\Big|\{\mathcal{G}_{\mathcal{N}_1,\mathcal{N}_2,\mathcal{N}_3}^{\mathbb{O}}(f,\Psi)\}(\omega_1,\omega_2,\omega_3)\Big|^2\\&\times{\rm d}\omega_1{\rm d}\omega_2{\rm d}\omega_3 {\rm d}\mu_1{\rm d}\mu_2{\rm d}\mu_3+ ||\Psi||^2_2 \int_{-\infty}^{\infty}\int_{-\infty}^{\infty}\int_{-\infty}^{\infty} \ln\Big|(t_1,t_2,t_3)\Big|\Big|f(t_1,t_2,t_3)\Big|^2{\rm d}t_1{\rm d}t_2{\rm d}t_3 \\&\geq K_{0}^{'}||\Psi||^2_2 \int_{-\infty}^{\infty}\int_{-\infty}^{\infty}\int_{-\infty}^{\infty}\Big|f(t_1,t_2,t_3)\Big|^2{\rm d}t_1{\rm d}t_2{\rm d}t_3.
	\end{align*}
	Hence, the desired result has been obtained.
\end{proof}
\begin{theorem}[Sharp Young-Hausdorff inequality for 3D WOCLCT] Let $1 \leq p<2$ and $q$ be such that $\frac{1}{p}+\frac{1}{q}=1$, then we have
	\begin{align*}
	||\{\mathcal{G}_{\mathcal{N}_1,\mathcal{N}_2,\mathcal{N}_3}^{\mathbb{O}}(f,\Psi)\}(\omega_1,\omega_2,\omega_3)||_{L^q(\mathbb{R}^3,\mathbb{O})} \leq A||f(t_1,t_2,t_3)||_{L^1(\mathbb{R}^3,\mathbb{O})}||\Psi||_{L^p(\mathbb{R}^3,\mathbb{O})},
	\end{align*}
	where $A = {(2 \pi)^{\frac{1}{q}-\frac{1}{p}-\frac{1}{2}}}{ |b_3|}^{-\frac{1}{2}}|b_1b_2|^{\frac{1}{q}-\frac{1}{2}}\left(\frac{p^{\frac{1}{p}}}{q^{\frac{1}{q}}}\right)$.
\end{theorem}
\begin{proof}
	From the  sharp Young--Hausdorff inequality for 3D OCLCT, as per our notations, is given by
	\begin{align}\label{eq:4.7}
	&\left(\int_{-\infty}^{\infty}\int_{-\infty}^{\infty}\int_{-\infty}^{\infty}\Big|f(t_1,t_2,t_3)\kappa_{\mathcal{N}_{1}}^{e_{1}}(t_{1},\omega_{1})\kappa_{\mathcal{N}_{2}}^{e_{2}}(t_{2},\omega_{2})\kappa_{\mathcal{N}_{3}}^{e_{4}}(t_{3},\omega_{3}){\rm d}t_1{\rm d}t_2{\rm d}t_3\Big|^q\right)^{\frac{1}{q}}\nonumber\\&\leq A \left(\int_{-\infty}^{\infty}\int_{-\infty}^{\infty}\int_{-\infty}^{\infty}\Big|f(t_1,t_2,t_3){\rm d}t_1{\rm d}t_2{\rm d}t_3\Big|^p\right)^{\frac{1}{p}}.
	\end{align}
	Now, replace  $f(t_1,t_2,t_3)$ by $f(t_1,t_2,t_3)\overline{\Psi(t_1-\mu_1,t_2-\mu_2,t_3-\mu_3)}$ in inequality \eqref{eq:4.7}, we obtain
	\begin{align}\label{eq:4.8}
	&	\Bigg(\int_{-\infty}^{\infty}\int_{-\infty}^{\infty}\int_{-\infty}^{\infty}\Big|f(t_1,t_2,t_3)\overline{\Psi(t_1-\mu_1,t_2-\mu_2,t_3-\mu_3)}\kappa_{\mathcal{N}_{1}}^{e_{1}}(t_{1},\omega_{1})\kappa_{\mathcal{N}_{2}}^{e_{2}}(t_{2},\omega_{2})\kappa_{\mathcal{N}_{3}}^{e_{4}}(t_{3},\omega_{3})\nonumber\\&\times{\rm d}t_1{\rm d}t_2{\rm d}t_3\Big|^q\Bigg)^{\frac{1}{q}} \leq A \left(\int_{-\infty}^{\infty}\int_{-\infty}^{\infty}\int_{-\infty}^{\infty}\left|f(t_1,t_2,t_3)\overline{\Psi(t_1-\mu_1,t_2-\mu_2,t_3-\mu_3)}{\rm d}t_1{\rm d}t_2{\rm d}t_3\right|^p\right)^{\frac{1}{p}}.
	\end{align}
	Now, using \cite[Theorem 1.3, pp.3]{book2} on the right-hand side of the inequality \eqref{eq:4.8}, we get
	\begin{align*}
	||\{\mathcal{G}_{\mathcal{N}_1,\mathcal{N}_2,\mathcal{N}_3}^{\mathbb{O}}(f,\Psi)\}(\omega_1,\omega_2,\omega_3)||_{L^q(\mathbb{R}^3,\mathbb{O})} \leq A||f(t_1,t_2,t_3)||_{L^1(\mathbb{R}^3,\mathbb{O})}||\Psi||_{L^p(\mathbb{R}^3,\mathbb{O})}.
	\end{align*}
\end{proof}
\begin{theorem}[Heisenberg's uncertainty principle for 3D WOCLCT] Suppose $f \in L^1(\mathbb{R}^3,\mathbb{O})\cap L^2(\mathbb{R}^3,\mathbb{O})$, then the following inequality is satisfied:
	\begin{align*}
	&\int_{-\infty}^{\infty}\int_{-\infty}^{\infty}\int_{-\infty}^{\infty}\Big|(t_1,t_2,t_3)\Big|^2\Big|f(t_1,t_2,t_3)\Big|^2{\rm d}t_1{\rm d}t_2{\rm d}t_3 \int_{-\infty}^{\infty}\int_{-\infty}^{\infty}\int_{-\infty}^{\infty}\int_{-\infty}^{\infty}\int_{-\infty}^{\infty}\int_{-\infty}^{\infty}\Big|(\omega_{1},\omega_{2},\omega_{3})\Big|^2\\&\times\Big|\{\mathcal{G}_{\mathcal{N}_1,\mathcal{N}_2,\mathcal{N}_3}^{\mathbb{O}}(f,\Psi)\}(\omega_1,\omega_2,\omega_3)\Big|^2{\rm d}\omega_{1}{\rm d}\omega_{2}{\rm d}\omega_{3}{\rm d}\mu_1{\rm d}\mu_2{\rm d}\mu_3\geq\frac{2}{\pi|b_3|}b_1^2b_2^2||f(t_1,t_2,t_3)||^2_2.
	\end{align*}
\end{theorem}
\begin{proof}
	From \cite{wentheoctonion2021}, Heisenberg's uncertainty principle for 3D OCLCT can be represented with our notations as follows:
	\begin{align}\label{eq:4.5}
	&\int_{-\infty}^{\infty}\int_{-\infty}^{\infty}\int_{-\infty}^{\infty}\Big|(t_1,t_2,t_3)\Big|^2\Big|f(t_1,t_2,t_3)\Big|^2{\rm d}t_1{\rm d}t_2{\rm d}t_3\int_{-\infty}^{\infty}\int_{-\infty}^{\infty}\int_{-\infty}^{\infty}\Big|(\omega_1,\omega_2,\omega_3)\Big|^2\nonumber \\&\times\Big|(\mathcal{L}_{\mathcal{N}_1,\mathcal{N}_2,\mathcal{N}_3}^{\mathbb{O}}f)(\omega_1,\omega_2,\omega_3)\Big|^2{\rm d}\omega_1{\rm d}\omega_2{\rm d}\omega_3\nonumber\\&\geq \frac{2}{\pi |b_3|}b_1^2b_2^2\int_{-\infty}^{\infty}\int_{-\infty}^{\infty}\int_{-\infty}^{\infty}\Big|f(t_1,t_2,t_3)\Big|^2{\rm d}t_1{\rm d}t_2{\rm d}t_3.
	\end{align}
	Now, replace  $f(t_1,t_2,t_3)$ by $f(t_1,t_2,t_3)\overline{\Psi(t_1-\mu_1,t_2-\mu_2,t_3-\mu_3)}$ and integrating with respect to $\mu_1,\mu_2$, and $\mu_3$ in equation \eqref{eq:4.5}, we obtain
	\begin{align*}
	&\int_{-\infty}^{\infty}\int_{-\infty}^{\infty}\int_{-\infty}^{\infty}\int_{-\infty}^{\infty}\int_{-\infty}^{\infty}\int_{-\infty}^{\infty}\Big|(t_1,t_2,t_3)\Big|^2\Big|f(t_1,t_2,t_3)\overline{\Psi(t_1-\mu_1,t_2-\mu_2,t_3-\mu_3)}\Big|^2\nonumber\\&\times{\rm d}t_1{\rm d}t_2{\rm d}t_3{\rm d}\mu_1{\rm d}\mu_2{\rm d}\mu_3\int_{-\infty}^{\infty}\int_{-\infty}^{\infty}\int_{-\infty}^{\infty}\int_{-\infty}^{\infty}\int_{-\infty}^{\infty}\int_{-\infty}^{\infty}\Big|(\omega_1,\omega_2,\omega_3)\Big|^2\Big|\{\mathcal{G}_{\mathcal{N}_1,\mathcal{N}_2,\mathcal{N}_3}^{\mathbb{O}}(f,\Psi)\}\\&(\omega_1,\omega_2,\omega_3)\Big|^2{\rm d}\omega_1{\rm d}\omega_2{\rm d}\omega_3{\rm d}\mu_1{\rm d}\mu_2{\rm d}\mu_3\geq \frac{2}{\pi |b_3|}b_1^2b_2^2\int_{-\infty}^{\infty}\int_{-\infty}^{\infty}\int_{-\infty}^{\infty}\int_{-\infty}^{\infty}\int_{-\infty}^{\infty}\int_{-\infty}^{\infty}\\&\Big|f(t_1,t_2,t_3)\overline{\Psi(t_1-\mu_1,t_2-\mu_2,t_3-\mu_3)}\Big|^2{\rm d}t_1{\rm d}t_2{\rm d}t_3{\rm d}\mu_1{\rm d}\mu_2{\rm d}\mu_3.
	\end{align*}
	\begin{align*}
	&||\Psi||_2^2\int_{-\infty}^{\infty}\int_{-\infty}^{\infty}\int_{-\infty}^{\infty}\Big|(t_1,t_2,t_3)\Big|^2\Big|f(t_1,t_2,t_3)\Big|^2{\rm d}t_1{\rm d}t_2{\rm d}t_3\int_{-\infty}^{\infty}\int_{-\infty}^{\infty}\int_{-\infty}^{\infty}\int_{-\infty}^{\infty}\int_{-\infty}^{\infty}\int_{-\infty}^{\infty}
	\\&\Big|(\omega_1,\omega_2,\omega_3)\Big|^2\Big|\{\mathcal{G}_{\mathcal{N}_1,\mathcal{N}_2,\mathcal{N}_3}^{\mathbb{O}}(f,\Psi)\}(\omega_1,\omega_2,\omega_3)\Big|^2{\rm d}\omega_1{\rm d}\omega_2{\rm d}\omega_3{\rm d}\mu_1{\rm d}\mu_2{\rm d}\mu_3\\&\geq \frac{2}{\pi |b_3|}b_1^2b_2^2||\Psi||^2_2||f(t_1,t_2,t_3)||^2_2
	\end{align*}
	\begin{align*}
	&\int_{-\infty}^{\infty}\int_{-\infty}^{\infty}\int_{-\infty}^{\infty}\Big|(t_1,t_2,t_3)\Big|^2\Big|f(t_1,t_2,t_3)\Big|^2{\rm d}t_1{\rm d}t_2{\rm d}t_3 \int_{-\infty}^{\infty}\int_{-\infty}^{\infty}\int_{-\infty}^{\infty}\int_{-\infty}^{\infty}\int_{-\infty}^{\infty}\int_{-\infty}^{\infty}\\&\Big|(\omega_{1},\omega_{2},\omega_{3})\Big|^2\Big|\{\mathcal{G}_{\mathcal{N}_1,\mathcal{N}_2,\mathcal{N}_3}^{\mathbb{O}}(f,\Psi)\}(\omega_1,\omega_2,\omega_3)\Big|^2{\rm d}\omega_{1}{\rm d}\omega_{2}{\rm d}\omega_{3}{\rm d}\mu_1{\rm d}\mu_2{\rm d}\mu_3\\&\geq\frac{2}{\pi|b_3|}b_1^2b_2^2||f(t_1,t_2,t_3)||^2_2.
	\end{align*}
	Hence, we get the desired result.
\end{proof}
\begin{theorem}[Donoho-Stark's uncertainty principle for 3D WOCLCT]Let $\sigma $ and $\tau$ be two measurable subset of $\mathbb{R}^3$ and  $f \in L^1(\mathbb{R}^3,\mathbb{O})\cap L^2(\mathbb{R}^3,\mathbb{O})$. If $f(t_1,t_2,t_3)$ is $\varepsilon_{\sigma}$- concentration to $\sigma$ in $L^1(\mathbb{R}^3,\mathbb{O})$- norm and $\{\mathcal{G}_{\mathcal{N}_1,\mathcal{N}_2,\mathcal{N}_3}^{\mathbb{O}}(f,\Psi)\}(\omega_1,\omega_2,\omega_3)$ is  $\varepsilon_{\tau}$- concentration to $\tau$ in $L^2(\mathbb{R}^3,\mathbb{O})$- norm, then
	\begin{align*}
	&\int_{-\infty}^{\infty}	\int_{-\infty}^{\infty}	\int_{-\infty}^{\infty}	\int_{-\infty}^{\infty}	\int_{-\infty}^{\infty}	\int_{-\infty}^{\infty}	\int_{-\infty}^{\infty}\Big|\{\mathcal{G}_{\mathcal{N}_1,\mathcal{N}_2,\mathcal{N}_3}^{\mathbb{O}}(f,\Psi)\}(\omega_1,\omega_2,\omega_3)\Big|^2{\rm d}\omega_1{\rm d}\omega_2{\rm d}\omega_3{\rm d}\mu_1{\rm d}\mu_2{\rm d}\mu_3\\&=\frac{|\sigma ||\tau|}{8 \pi^3|b_1b_2b_3|[(1-\varepsilon_{\sigma})(1-\varepsilon_{\tau})]^2}||\Psi||^2_2||f(t_1,t_2,t_3)||^2_2.
	\end{align*}
\end{theorem}
\begin{proof}
	From \cite{wentheoctonion2021},	the Donoho-Stark's uncertainty principle for 3D OCLCT can be represented with our notations as follows:
	\begin{align}\label{eq:4.6}
	&\int_{-\infty}^{\infty}\int_{-\infty}^{\infty}\int_{-\infty}^{\infty}\Big|(\mathcal{L}^{\mathbb{O}}_{\mathcal{N}_{1},\mathcal{N}_{2},\mathcal{N}_3}f)(\omega_1,\omega_2,\omega_3)\Big|^2{\rm d}\omega_1{\rm d}\omega_2{\rm d}\omega_3\nonumber\\& =\frac{|\sigma ||\tau|}{8 \pi^3|b_1b_2b_3|[(1-\varepsilon_{\sigma})(1-\varepsilon_{\tau})]^2}\int_{-\infty}^{\infty}\int_{-\infty}^{\infty}\int_{-\infty}^{\infty}\Big|f(t_1,t_2,t_3)\Big|^2{\rm d}t_1{\rm d}t_2{\rm d}t_3.
	\end{align}
	Now replace  $f(t_1,t_2,t_3)$ by $f(t_1,t_2,t_3)\overline{\Psi(t_1-\mu_1,t_2-\mu_2,t_3-\mu_3)}$ and integrating with respect to $\mu_1,\mu_2$, and $\mu_3$ in equation \eqref{eq:4.6}, we obtain
	\begin{align*}
	&\int_{-\infty}^{\infty}	\int_{-\infty}^{\infty}	\int_{-\infty}^{\infty}	\int_{-\infty}^{\infty}	\int_{-\infty}^{\infty}	\int_{-\infty}^{\infty}	\int_{-\infty}^{\infty}\Big|\{\mathcal{G}_{\mathcal{N}_1,\mathcal{N}_2,\mathcal{N}_3}^{\mathbb{O}}(f,\Psi)\}(\omega_1,\omega_2,\omega_3)\Big|^2{\rm d}\omega_1{\rm d}\omega_2{\rm d}\omega_3{\rm d}\mu_1{\rm d}\mu_2{\rm d}\mu_3\\&=\frac{|\sigma ||\tau|}{8 \pi^3|b_1b_2b_3|[(1-\varepsilon_{\sigma})(1-\varepsilon_{\tau})]^2}\int_{-\infty}^{\infty}\int_{-\infty}^{\infty}\int_{-\infty}^{\infty}\int_{-\infty}^{\infty}\int_{-\infty}^{\infty}\int_{-\infty}^{\infty}\Big|f(t_1,t_2,t_3)\\&\times\overline{\Psi(t_1-\mu_1,t_2-\mu_2,t_3-\mu_3)}\Big|^2{\rm d}t_1{\rm d}t_2{\rm d}t_3{\rm d}\mu_1{\rm d}\mu_2{\rm d}\mu_3=\frac{|\sigma ||\tau|}{8 \pi^3|b_1b_2b_3|[(1-\varepsilon_{\sigma})(1-\varepsilon_{\tau})]^2}\\&\times\int_{-\infty}^{\infty}\int_{-\infty}^{\infty}\int_{-\infty}^{\infty}\int_{-\infty}^{\infty}\int_{-\infty}^{\infty}\int_{-\infty}^{\infty}\Big|f(t_1,t_2,t_3)\Big|^2|{\Psi(t_1-\mu_1,t_2-\mu_2,t_3-\mu_3)}|^2\\&\times{\rm d}t_1{\rm d}t_2{\rm d}t_3{\rm d}\mu_1{\rm d}\mu_2{\rm d}\mu_3=\frac{|\sigma ||\tau|}{8 \pi^3|b_1b_2b_3|[(1-\varepsilon_{\sigma})(1-\varepsilon_{\tau})]^2}||\Psi||^2_2||f(t_1,t_2,t_3)||^2_2.
	\end{align*}
\end{proof}
\section{Potential applications of 3D WOCLCT}\label{sec:5}
As discussed in the introduction, the real-life applications of hyper-complex algebra-based transforms are important tools in the modern age or the cutting-edge science and engineering field. Many scientists have paid much attention to the applications of octonions, such as structural design,
predicting earthquakes using seismic signals, computer graphics, aerospace engineering, quantum mechanics,
time-frequency analysis, optics, signal processing, image processing and enhancement, pattern recognition,
artificial intelligence, etc. Many applications in the literature use OCLCT, which deals with only multi-channel stationary signals and is not compatible with multi-channel non-stationary signals. In contrast, the 3D WOCLCT is an important and relevant tool that deals with multi-channel stationary and non-stationary signals. In a real-life application, the advantage of this WOCLCT tool is to handle signals by using a window function that simultaneously localizes the hyper-complex signal in the time and frequency domain. 
\section{Conclusion}\label{sec:6}
In this work, we investigated octonion algebra using window linear canonical transform (WLCT), defined as WOCLCT. We have provided a new definition of 3D WOCLCT and constructed its inversion formula. Following the present technique, we have derived many interesting properties (linearity, parity, and shifting properties) of 3D WOCLCT, including a relationship between 3D WOCLCT and OCLCT. Further, the main contribution of this work is to obtain inequalities and uncertainty principles (such as sharp Pitt's inequality, Sharp Young-Hausdorff inequality, logarithmic uncertainty principle, Heisenberg's and Donoho-Stark's uncertainty principles for 3D WOCLCT) are derived. The potential applications of 3D WOCLCT are also discussed. The results obtained in this paper are presumably new and very useful in theoretical and mathematical physics, including signal processing and optics.
\section*{Acknowledgments} The second-named author is grateful to BITS-Pilani, Hyderabad Campus, for providing research funds (ID No. 2022PHXP0424H).
\section*{Declaration of competing interest}
The authors declare that they have no known competing financial interests or personal relationships that could have appeared to influence the work reported in this paper. 
\section*{Data availability}
No data was used for the research described in the article. 


\begin{thebibliography}{20}
	\bibitem{color}C.J. Evans, S.J. Sangwine, T.A. Ell, Colour-sensitive edge detection using hypercomplex filters. In 2000 10th European Signal Processing Conference (pp. 1-4) (2000) IEEE.
	\bibitem{novel}C. Gao, J. Zhou, F. Lang, Q. Pu, C. Liu, A novel approach to edge detection of color image based on quaternion fractional directional differentiation. In Advances in Automation and Robotics, Vol. 1: Selected Papers from the 2011 International Conference on Automation and Robotics (ICAR 2011), Dubai, December 1–2 (2012) 2011 (pp. 163-170). Springer Berlin Heidelberg.
	\bibitem{pattern} S.C. Pei, J.J. Ding, J. Chang, Color pattern recognition by quaternion correlation. In Proceedings 2001 International Conference on Image Processing (Cat. No. 01CH37205)  (2001, October) (Vol. 1, pp. 894-897). IEEE.
	\bibitem{image}Zhou, Heng, Chunlei Zhang, Xin Zhang, Qiaoyu Ma, Image classification based on quaternion-valued capsule network. Applied Intelligence 53, no. 5 (2023) 5587-5606.
	\bibitem{two}Li, Zhen-Wei, Bing-Zhao Li, and Min Qi, Two-dimensional quaternion linear canonical series for color images. Signal Processing: Image Communication 101 (2022) 116574.
	\bibitem{hyper}T.A. Ell, S.J. Sangwine, Hypercomplex Fourier transforms of color images. IEEE Transactions on image processing, 16(1) (2006) 22-35.
	
	\bibitem{akQWLCT2023} A. Prasad, M. Kundu, Uncertainty principles and applications of quaternion windowed linear canonical transform. Optik, 272 (2023) 170220.
	\bibitem{1} M. Bahri, Quaternion linear canonical transform application. Global Journal of Pure and Applied Mathematics, 11 (2015) 19-24.
	\bibitem{bayrofourier2007} E. Bayro-Corrochano, N. Trujillo, M. Naranjo, Quaternion Fourier descriptors for the preprocessing and recognition of spoken words using images of spatiotemporal representations. Journal of Mathematical Imaging and Vision, 28 (2007) 179-190.
	\bibitem{pbaswater2003} P. Bas, N. Le Bihan, J.M. Chassery,  Color image watermarking using quaternion Fourier transform. In 2003 IEEE International Conference on Acoustics, Speech, and Signal Processing, 2003. Proceedings. (ICASSP'03). IEEE(Vol. 3 (2003) pp. III-521). 
	\bibitem{2} T.A. Ell, N. Le Bihan, S.J. Sangwine,  Quaternion Fourier transforms for signal and image processing. John Wiley \& Sons, 2014.
	\bibitem{3} X. Guanlei, W. Xiaotong, X. Xiaogang, Fractional quaternion Fourier transform, convolution and correlation. Signal Processing, 88 (2008), 2511-2517.
	\bibitem{panoctonion2019} P. Lian,  The octonionic Fourier transform: Uncertainty relations and convolution. Signal Processing, 164 (2019) 295-300.
	\bibitem{WLCT} K.I. Kou, R.H. Xu, Windowed linear canonical transform and its applications. Signal Processing, 92 (2012) 179-188.
	\bibitem{havinuncer2012} V. Havin, B. Jöricke, The uncertainty principle in harmonic analysis (Vol. 28). Springer Science \& Business Media, 2012.
	\bibitem{follandmathematical1997} G.B. Folland, A. Sitaram, The uncertainty principle: a mathematical survey. Journal of Fourier analysis and applications, 3 (1997) 207-238.
	\bibitem{williampitt} W. Beckner, Pitt’s inequality and the uncertainty principle. Proceedings of the American Mathematical Society, 123 (1995) 1897-1905.
	\bibitem{wentheoctonion2021} W.B. Gao, B.Z. Li, The octonion linear canonical transform: Definition and properties. Signal processing, 188 (2021) 108233.
	\bibitem{book} T. Dray, C. A. Manogue, The geometry of the octonions. World Scientific, 2015.
	\bibitem{book2}E.M. Stein, G. Weiss, Introduction to Fourier analysis on Euclidean spaces (Vol. 1). Princeton university press, 1971.
	\bibitem{app} Ł. Błaszczyk, A generalization of the octonion Fourier transform to 3-D octonion-valued signals: properties and possible applications to 3-D LTI partial differential systems. Multidimensional Systems and Signal Processing, 31(2020) 1227-1257.
\end{thebibliography}
\end{document}